\documentclass[11pt,letterpaper]{article}
\usepackage[margin=1in]{geometry}
\usepackage[T1]{fontenc}
\usepackage{lmodern}
\usepackage{graphicx}
\usepackage{amsmath,amssymb,amsthm,mathtools,bm,xfrac}
\usepackage[dvipsnames]{xcolor}
\usepackage{enumerate}
\usepackage{algorithm,algpseudocode,tikz,subfig}
\usepackage{natbib}
\usepackage{microtype}
\usepackage{xurl}
\usepackage[hidelinks]{hyperref}
\usepackage[capitalize]{cleveref}
\makeatletter
\def\arxivrefkind#1:#2\@nil{\csname arxivrefkind@#1\endcsname}
\def\arxivrefkind@sec{Section}
\def\arxivrefkind@thm{Theorem}
\def\arxivrefkind@lem{Lemma}
\def\arxivrefkind@prop{Proposition}
\def\arxivrefkind@corollary{Lemma}
\makeatother

\newtheorem{theorem}{Theorem}
\newtheorem{lemma}{Lemma}
\newtheorem{proposition}{Proposition}

\newtheorem{assumption}{Assumption}
\theoremstyle{definition}
\newtheorem{definition}{Definition}

\theoremstyle{remark}

\makeatletter
\renewenvironment{proof}[1]{\par\pushQED{}\normalfont
  \topsep6\p@\@plus6\p@\relax
  \trivlist\item[\hskip\labelsep\itshape #1]\ignorespaces
}{\popQED\endtrivlist\@endpefalse}
\makeatother

\usepackage{ifthen}

\newcommand{\given}{\,|\,}

\newcommand{\prob}[2][]{\text{\bf Pr}\ifthenelse{\not\equal{}{#1}}{_{#1}}{}\!\left[{\def\givenn{\middle|}#2}\right]}
\newcommand{\expect}[2][]{\mathbb{E}\ifthenelse{\not\equal{}{#1}}{_{#1}}{}\!\left[{\def\givenn{\middle|}#2}\right]}

\newcommand{\expecto}[2][]{\mathbb{E}\ifthenelse{\not\equal{}{#1}}{^{#1}}{}\!\left[{\def\givenn{\middle|}#2}\right]}

\newcommand{\tparen}{\big}
\newcommand{\tprob}[2][]{\text{\bf Pr}\ifthenelse{\not\equal{}{#1}}{_{#1}}{}\tparen[{\def\given{\tparen|}#2}\tparen]}
\newcommand{\texpect}[2][]{\text{\bf E}\ifthenelse{\not\equal{}{#1}}{_{#1}}{}\tparen[{\def\given{\tparen|}#2}\tparen]}

\newcommand{\sprob}[2][]{\text{\bf Pr}\ifthenelse{\not\equal{}{#1}}{_{#1}}{}[#2]}
\newcommand{\sexpect}[2][]{\text{\bf E}\ifthenelse{\not\equal{}{#1}}{_{#1}}{}[#2]}

\definecolor{estTruthful}{HTML}{3366BB}
\definecolor{estNaive}{HTML}{C44E52}
\definecolor{estSameDay}{HTML}{339966}
\definecolor{estUnknown}{HTML}{DD8822}

\newcommand{\ind}[1]{\mathbf{1}[#1]}

\newcommand{\mb}[1]{\ensuremath{\boldsymbol{#1}}}

\newcommand{\price}[1][]{p\ifx#1\empty\else^{(#1)}\fi}

\newcommand{\countsameday}[1][]{N\ifx#1\empty\else^{(#1)}\fi}

\newcommand{\patience}{\gamma}

\newcommand{\low}{\epsilon}
\newcommand{\lowprob}{q}

\newcommand{\pz}{p^{(0)}}
\newcommand{\pstar}{p_*}

\newcommand{\edd}{\hat{d}}

\newcommand{\eddunbias}{\hat{d}}

\newcommand{\bias}{\mathcal{B}}

\newcommand{\mcE}{\ensuremath{\mathcal E}}
\newcommand{\mcQ}{\ensuremath{\mathcal Q}}

\newcommand{\rev}{\text{Rev}}

\newcommand{\drev}{r}

\bibpunct[, ]{(}{)}{,}{a}{}{,}
\title{Switchback Price Experiments with Forward-Looking Demand}
\author{Yifan Wu\thanks{Microsoft Research. E-mail: \texttt{yifan.wu2357@gmail.com}}
  \and Ramesh Johari\thanks{Management Science and Engineering, Stanford University. E-mail: \texttt{rjohari@stanford.edu}}
  \and Vasilis Syrgkanis\thanks{Management Science and Engineering, Stanford University. E-mail: \texttt{vsyrgk@stanford.edu}}
  \and Gabriel Y. Weintraub\thanks{FEN, Universidad de Chile. E-mail: \texttt{gweintra@fen.uchile.cl}}}
\date{}
\hypersetup{pdftitle={Switchback Price Experiments with Forward-Looking Demand},
  pdfauthor={Yifan Wu, Ramesh Johari, Vasilis Syrgkanis, Gabriel Y. Weintraub}}

\begin{document}
\maketitle
\begin{abstract}
Switchback experiments are widely used in pricing applications. We consider a retailer running such an experiment for the price of a single product, with infinite supply. In each period, the seller chooses a price from a set of predefined prices consisting of an original price and a few discounted price levels. The goal is to estimate the demand gradient at the original price to adjust the original price and improve revenue after the experiment. In our model, in each period, a unit mass of buyers arrives on the market, with values distributed based on a time-varying process. Crucially, buyers are forward-looking with a discounted utility and will choose not to purchase now if they expect to face a discounted price in the near future. We show that such forward-looking behavior introduces bias in naive estimators of the demand gradient, due to intertemporal interference. Furthermore, although experiments with only two price points are common in practice, we prove that no estimator using data from such experiments can recover the correct demand gradient, even in the limit of an infinitely long experiment with an infinitesimal price discount. Moreover, we characterize the form of the bias of naive estimators. Finally, we show that with a simple three-price-level experiment, the seller can remove the bias due to strategic forward-looking behavior and construct an estimator for the demand gradient that asymptotically recovers the truth. Our results provide practical guidance for platforms to obtain unbiased demand gradients from price experiments in the presence of forward-looking demand.
\end{abstract}

\section{Introduction}
\label{sec: intro}

Online platforms routinely conduct thousands of A/B tests to evaluate interventions ranging from user interface changes to pricing strategies to algorithm modifications \citep{kohavi2020trustworthy}.\footnote{This paper was accepted at the 27th ACM Conference on Economics and Computation (EC'26) and appeared in the proceedings as an extended abstract \citep{anonymous2026switchback}.} A fundamental premise underlying standard A/B testing is the {\em stable unit treatment value assumption} (SUTVA) \citep{imbens2015causal}: treatment assignment affects a unit's own potential outcomes, but not the potential outcomes of {\em other} units.  Violations of SUTVA (known as {\em interference}) lead to biases in the standard designs and estimators used in A/B testing.  Extensive recent literature has developed approaches to estimating and mitigating these biases in platform experiments (see Section \ref{sec:related} for a brief review).

A particularly consequential form of interference arises when the units in the experiment are {\em strategic agents} who optimize their decisions in response to the treatment regime.  For example, a seller might change how they set prices; a marketplace platform might change the ranking of participants in search results; or a matching platform might change the algorithm used to make matches. In such settings,
the experimenter's goal is typically to estimate the difference between two benchmark {\em equilibrium outcomes}: the equilibrium when the whole system is in treatment ({\em global treatment}) vs.~control ({\em global control}).  The difference in these quantities is the {\em global treatment effect} (GTE).

When we run experiments in such settings, we do not observe either global treatment or global control perfectly.  For example, a standard Bernoulli randomized experimental design would lead to a subset of strategic agents experiencing treatment, and others experiencing control.  It is far from obvious that such ``mixtures'' of strategic agent exposures give us the information we need to obtain unbiased estimates of the GTE.  Hence, a central practical question that has emerged is: {\em When can we use data observed from strategic agents during an experiment, to reason about the GTE?}

Recently, this challenge has received significant attention at the intersection of causal inference, economics, and operations.  For example, \citet{wager2021experimenting} and \citet{munro2021treatment} study biases present when outcomes arise due to market equilibrium (the former in a particular setting of a two-sided market, and the latter in a more general competitive equilibrium framework); both papers suggest how to correct these biases using local price experimentation. Other papers have studied the impact of strategic behavior in settings where the entire mechanism changes, e.g., in matching \citep{arkhangelsky2025evaluating}, as well as policy optimization in the presence of strategic behavior \citep{munro2025causal}.  \citet{chawla2016b} presents an analysis and debiasing design for A/B testing of auction mechanisms, in settings where bidders may be uncertain of which auction mechanism their bids are entered into (as in advertising auctions).  For further discussion, see Section \ref{sec:related}.

A crucial aspect of all of these works is that they consider {\em static} strategic equilibria.  In other words, the experimental data collected is cross-sectional rather than intertemporal or dynamic; and the goal is to estimate the static equilibrium behavior in global treatment and global control respectively (or more precisely, the difference between them, as represented by the GTE).  By exploiting structural relationships between global treatment, global control, and the experimental system, these papers are able to both characterize bias and suggest means to reduce or eliminate it.

The central contribution of our paper involves studying experimentation in the presence of strategic behavior, but in a {\em dynamic} setting.  A dynamic environment also opens up the possibility of experimental designs that vary treatment allocation over time.  In particular, a standard experimental design in such settings is the {\em switchback experiment}, where the entire system is switched between treatment and control over time \citep{bojinov2023design,brandt1938tests,lyft-switchback,hu2022switchback,kastelman2018switchback}.  In our paper, we study whether and how a platform can use switchback designs to  estimate global treatment effects with low bias in settings with intertemporal strategic behavior.

We focus on {\em price experimentation}, where a seller conducts switchback experiments by alternating between prices over time to estimate the gradient of demand with respect to price.  Ideally, if the seller can learn this demand gradient accurately, they can use the information to move prices in a revenue-maximizing direction.  However, when customers can observe or anticipate this experimental price variation, they may strategically delay purchases in anticipation of future discounts. This strategic behavior changes the experimental outcomes relative to the global treatment and global control equilibria, where prices are constant and no such strategic waiting occurs.

To make this concrete, consider a seller conducting a standard two-price switchback experiment to estimate the demand gradient at an original price $p$. Typically, the seller divides time into intervals and in each interval randomly sets price to either $p$ (control) or $p - \epsilon$ (treatment), where $\epsilon > 0$ is small. The seller then compares observed demand at the two prices to estimate the demand gradient. \Cref{fig:keepa} illustrates such price switching using price histories from Keepa (\url{https://keepa.com}). Crucially, the same price trackers that enable sellers to conduct experiments also enable customers to learn about discounts and behave strategically: a customer arriving when price is $p$ may wait rather than purchase, anticipating possible future discounts to $p - \epsilon$. This forward-looking behavior inflates observed treatment demand and deflates observed control demand because customers wait for lower prices, which in turn biases the estimated demand gradient. If severe enough, this estimation bias can lead the seller to move prices in the wrong direction, decreasing revenue. Importantly, this bias arises precisely because the switchback experiment induces different strategic behavior than global treatment (price fixed at $p-\epsilon$) or global control (price fixed at $p$), where customers have no incentive to wait.

\begin{figure}[htpb]
    \centering
    \begin{minipage}[t]{0.49\linewidth}
        \centering
        \includegraphics[width=\linewidth,trim=2bp 0bp 1bp 4bp,clip]{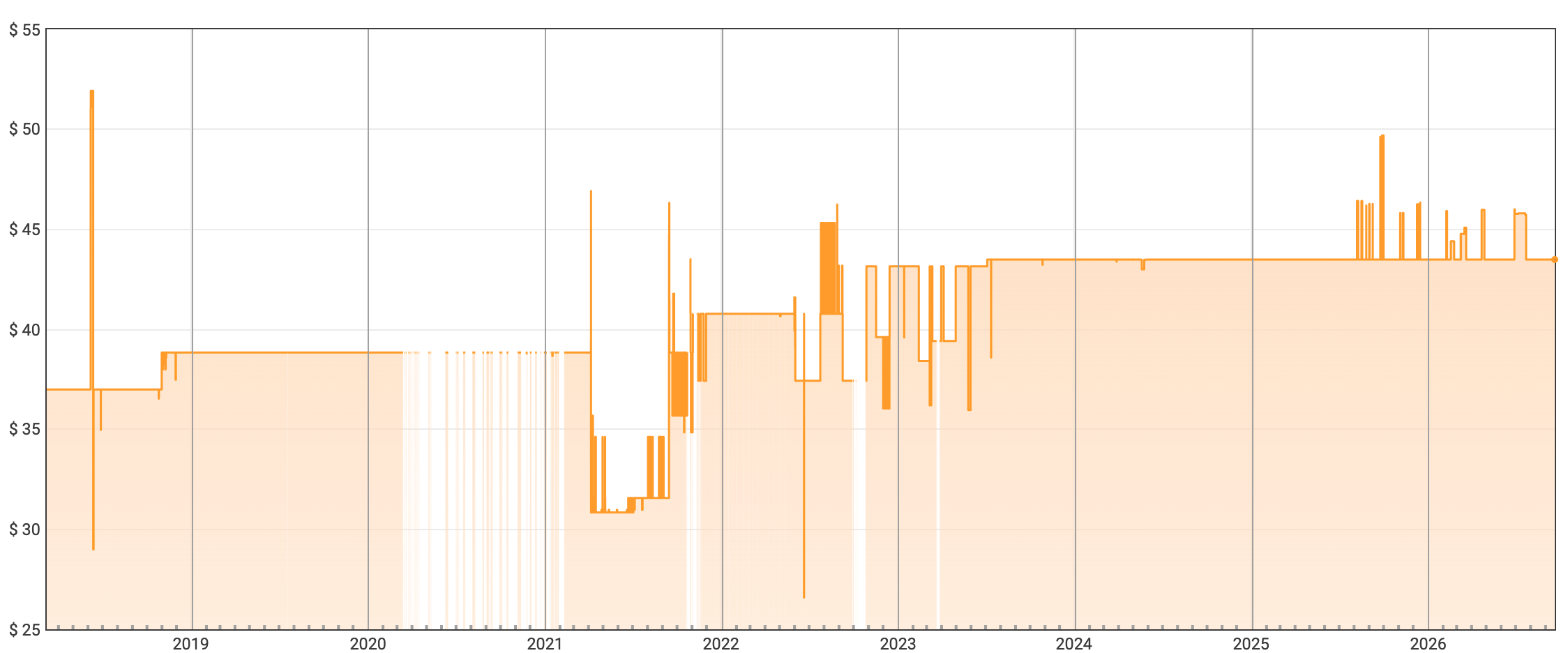}\\[2pt]
        {\small (a) Long-term price history}
    \end{minipage}\hfill
    \begin{minipage}[t]{0.49\linewidth}
        \centering
        \includegraphics[width=\linewidth,trim=0bp 1.5bp 2bp 2bp,clip]{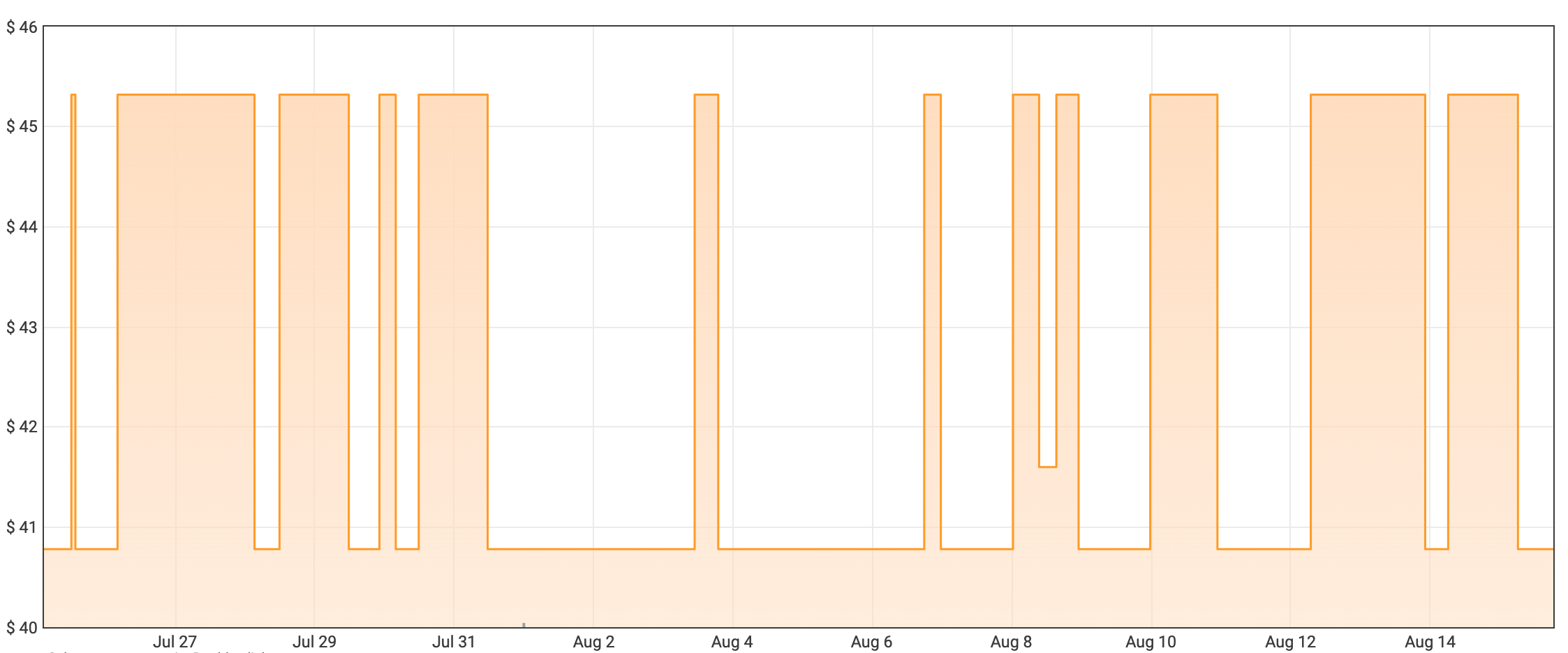}\\[2pt]
        {\small (b) Switching prices from July to August 2021. }
    \end{minipage}
    \vspace{3mm}
    \caption{Price history of a paper towel product on Amazon (ASIN: B079VP6DH5). The left panel shows the long-term price history; the right panel shows switching between price levels over a shorter time interval. The panels use different time and price scales. Source: Keepa (\protect\url{https://keepa.com}).}
    \label{fig:keepa}
\end{figure}

Our paper is about understanding, quantifying, and removing this bias.  We make three main contributions:
\begin{enumerate}
\item {\em We introduce a succinct yet general model in which we quantify how forward-looking behavior impacts estimation bias.} Formally, we quantify how the bias depends on the {\em discount factor} of arriving customers, i.e., their relative patience level to wait for future lower prices.
\item {\em We show that if restricted to the data from a standard two-price switchback experiment, commonly used in practice, this bias cannot be removed.}  Formally, we show that the true demand gradient is {\em unidentifiable} if customers are forward-looking and the standard switchback design above is used.
\item {\em We show how a simple modification to the switchback design that uses \textbf{three} prices, and an associated simple estimator, can completely remove this bias.}
\end{enumerate}

To carry out our analysis, we introduce a benchmark strategic discrete-time model of forward-looking behavior.  In the model we consider, the seller faces an arrival process of forward-looking customers who form {\em rational expectations} with respect to the future price trajectory they expect to see.  In our simplest model, all customers discount future rewards with a discount factor $\gamma < 1$. We also consider a generalization with heterogeneous discount factors.  Individual customers have unit demand, and a private value drawn from an unknown distribution; their only decision is whether or not to purchase the good in a given period at the current price, or wait.  In global treatment and global control, forward-looking behavior has no effect: if buying is rational, they will buy immediately.  But in the experiment, customers may strategically wait as discussed above.
To enable tractable analysis, we consider experiments with a random horizon length, with geometric length of mean $1/\delta$; we carry out our analysis in the limit where $\delta \to 0$, that is, when the expected experiment length grows large.

Our first and second results use this characterization to analyze the bias of the standard two-price switchback estimator. In the \textbf{first} result, we show that the estimator equals the true demand gradient multiplied by a factor that depends on the private patience parameter $\gamma$, and on whether the seller can distinguish ``same-day'' sales (from new arrivals) from delayed sales (from prior arrivals). For homogeneous buyer types, if the seller can distinguish these categories, the factor is $1+\frac{\gamma}{1-\gamma}q$, where $q$ is the probability that the discounted price is posted.  When the seller cannot distinguish same-day sales, the factor increases to $\frac{1}{q}+\frac{\gamma}{1-\gamma}$. \Cref{sec:warmup} provides an intuitive calculation of the bias.

While this might suggest that the seller can simply rescale the estimator to remove bias, our \textbf{second} result shows a fundamental obstacle: if $\gamma$ is unknown, then the true demand gradient is not identifiable from data generated by a two-price switchback experiment (even when same-day and delayed sales are observable). Intuitively, high observed delay can be explained either by high price sensitivity with low patience, or by low price sensitivity with high patience; thus a range of pairs of demand gradient and patience parameter $\gamma$ together can induce the same distribution over experimental outcomes.

Our \textbf{third} contribution resolves this issue via a minimal design change: a \emph{three} price experiment and a simple estimator eliminate estimation bias due to forward-looking behavior, under a general setting with heterogeneous and unknown patience that may be arbitrarily correlated with values. Specifically, the seller randomizes among an original price $p$ and two treatment prices $p-\epsilon$ and $p-2\epsilon$ (each treatment chosen with probability $q$, and the original price chosen with probability $1-2q$); for exposition, suppose the seller can distinguish same-day and delayed sales. Let $\Delta_1$ be the difference in average same-day sales between periods priced at $p-\epsilon$ and $p-2\epsilon$, and let $\Delta_2$ be the difference in average same-day sales between periods priced at $p$ and $p-\epsilon$. Absent forward-looking behavior, both are discrete estimates of the same demand gradient and should agree to first order. %
We show that, as $\epsilon \to 0$ and $\delta \to 0$, the estimator
\begin{equation}\label{eqn:intro-debias}
\hat{d} \;=\; \frac{\Delta_1 - (\Delta_2 - \Delta_1)}{\epsilon}
\end{equation}
is asymptotically unbiased for the true demand gradient: the difference-of-differences $\Delta_2-\Delta_1$ isolates the  first order  distortion induced by forward-looking behavior and removes it.  The resulting three-price switchback design is simple to implement and remains valid in substantially more general settings, including when same-day versus delayed sales are unobservable.
We note that in the case where only total sales are observed, the estimator takes a slightly different (but also simple) form from \Cref{eqn:intro-debias}.

Together, these contributions characterize the bias caused by forward-looking behavior and show how a three-price experiment can remove it asymptotically. We focus on debiasing for several reasons. First, the naive estimator can reverse the sign of the implied revenue gradient, leading the seller to decrease prices when an increase would improve revenue. A similar concern arises in price experiments with substitution across products, as shown by \citet{page2023price}.
Second, a seller may use experimental estimates in a learning algorithm that repeatedly adjusts prices to improve revenue. Standard convergence guarantees for stochastic gradient methods rely on unbiased gradient estimates \citep{nemirovski2009robust}, and unbiased reward estimates similarly play a central role in the regret analysis of bandit algorithms such as Exp3 \citep{auer2002nonstochastic}. These results provide a further motivation for removing the bias of the naive estimator when experimental estimates are used to guide subsequent pricing decisions.
Third, sampling variance can be driven arbitrarily close to zero by increasing the effective sample size, whereas bias from forward-looking behavior need not vanish with more data.
It is worth noting that, at a fixed sample size, removing bias can increase variance and need not improve estimation accuracy.  Our simulations in \Cref{sec: simulation static parameter} illustrate this bias-variance tradeoff: the unbiased three-price estimator can yield lower mean squared error when the bias is large, while the naive estimator can have lower mean squared error when its lower variance outweighs its bias.

\subsection{Related Work}
\label{sec:related}

\paragraph{Forward-Looking Demand and Intertemporal Substitution} Demand estimation in the presence of forward-looking demand and intertemporal substitution has a long history in industrial organization and econometrics \citep{aguirregabiria2013recent,akca2015identifying,birge2021markdown,blundell1989microeconometric,chen2020pricing,ching2020identification,diewert1974intertemporal,eichenbaum1990estimating,gandhi2021empirical,hartmann2006intertemporal,hendel2006measuring,hendel2013intertemporal,nevo2001measuring,nevo2002manufacturers,nevo2016usage}. Another line of work in operations management considers consumers' forward-looking behavior, mostly from a modeling perspective (e.g., \citet{su2007intertemporal,swinney2011selling}, etc.), but also from an empirical one \citep{li2014consumers}.

Even though, \citet{su2007intertemporal} and \citet{li2014consumers} do not study experimentation, they are closest to our paper because of their analysis of strategic forward-looking behavior.  \citet{su2007intertemporal} designs   a revenue-maximizing pricing policy with a continuum of buyers. Unlike our work which makes no assumption on buyers' private type space, \citet{su2007intertemporal} considers  a binary setting (i.e., binary buyer values and binary buyer patience levels), under which a characterization of the optimal policy is possible. \citet{li2014consumers} focuses on a setting with observational data, which requires more complex assumptions to achieve identification; in contrast, we focus on an experimental setting, and the demand gradient is identified with fewer assumptions.

\paragraph{Causal Inference, Interference, and Equilibrium}  Several recent papers study market experiments, and in particular interference in pricing experiments
\citep{blake2014marketplace,farias2022markovian,holtz2020reducing,johari2022experimental,le2023price,li2023experimenting}.

Our work is closest to two lines of work.  First is those that consider causal inference in the presence of equilibrium effects, as discussed in the introduction (see references therein).  We note further this line of work is closely related to the {\em Lucas critique} in macroeconomics \citep{lucas1976econometric}, which cautions that behavioral relationships (i.e., decision rules) estimated from historical data are generally not invariant to changes in policy -- they will change when government policy changes.

Our work is also closely related to prior work that specifically studies price experimentation; the closest of these to our paper is
\citet{page2023price}. \citet{page2023price} present an analysis of estimation bias due to interference stemming from substitution effects across similar products in a pricing experiment. They consider a \textit{single period} of experimentation, with applications to treatment on both sides of a two-sided market. The bias in their model can be corrected through a single experiment, supplemented by pre-experimental data.
In contrast, we examine a scenario where interference emerges \textit{across time periods}, and customers can intertemporally substitute when the design varies prices over time, and demand can be time-varying.

\paragraph{No-regret reserve price optimization with strategic buyers} Our work is also related to the line of work on dynamic pricing and regret minimization with strategic buyers \citep{chen2018robust,chen2022bayesian,feldman2016online,liu2024contextual,mashiah2023learning,mohri2014learning,mohri2014optimal,mohri2015revenue,mohri2016learning,tang2017reinforcement}. These works typically assume a stationary i.i.d. stochastic demand. The key intuition of the solutions underlying most of the work in this line is that a seller should change prices infrequently, so that buyers are not patient enough to anticipate the price discount that will occur when the price gets updated. However, in a constantly changing demand system, this can lead to very infrequent adaptation and inaccurate learning. Moreover, different to these works, our goal is in the analysis of a widely used practice in price experimentation, namely switchback experiments. Perhaps closest to our work is \citet{chen2018robust}, which considers online revenue maximization with randomly arriving customers. Despite the differences discussed above, \citet{chen2018robust} also considers selling to forward-looking buyers who discount future utility with a patience level.  They optimize revenue, while we focus on accurate gradient estimation.

\paragraph{Dynamic pricing and demand learning}
The proliferation of online e-commerce platforms has afforded sellers an increasingly rich plethora of algorithmic approaches to optimize the pricing of their products; see, e.g., \citep{mckinsey,bain,caro2012clearance,chen2016empirical}; even brick-and-mortar retail has benefited from these advances with the introduction of digital price tags
\citep{nicas2015now}.
Numerous modern algorithmic dynamic pricing methods are founded on the ability to precisely estimate this demand gradient, ranging from simple gradient ascent to more complex policy gradient methods using deep reinforcement learning \citep{liu2019dynamic}.

The line of work in operations research on dynamic pricing with an unknown demand is most related to our work \citep{agrawal2021dynamic,keskin2014dynamic,keskin2017chasing}.  Compared to \citet{agrawal2021dynamic, keskin2014dynamic} which make structural assumptions about the demand and  \citet{keskin2017chasing} which assumes the demand variation is bounded, our work captures more general demand-generating processes only with assumptions that they are locally approximable.
\section{A Warm-Up: Stationary Demand, Homogeneous Buyers, and an Infinite Horizon}
\label{sec:warmup}

Before presenting the general model, we first work through a simple stationary setting that delivers the main economic messages in the paper.

The warm-up in this section is deliberately specialized: demand is stationary, all buyers have the same patience level, and the experiment has an infinite horizon. In this setting, the bias created by forward-looking buyers can be derived with simple calculations. The model in \Cref{sec:model} is fully general in the sense that it allows heterogeneous buyers, nonstationary demand, and an experiment that ends with some probability in each period.

\subsection{A Two-Price Switchback Experiment}
\label{sec:warmup-two-price}

Consider an experiment that randomizes between an original price $p$ and a discounted price $p-\epsilon$, where $\epsilon>0$. In each period, the discounted price is posted independently with probability $q$, and the original price is posted with probability $1-q$. A unit mass of buyers arrives each period. Buyer values are drawn from a stationary distribution with demand function
\[
    D(x)=\Pr(v\geq x),
\]
and all buyers have the same patience level $\gamma\in[0,1)$. We focus on the infinite-horizon approximation, so buyers expect the same price randomization to continue forever.

If a buyer arrives when the discounted price $p-\epsilon$ is posted, there is no lower price in the experiment to wait for. Thus the buyer purchases immediately if and only if
\[
    v\geq p-\epsilon.
\]
The distortion occurs when a buyer arrives at the original price $p$. Buying immediately gives utility
\[
    u_{\mathrm{buy}}=v-p.
\]
If the buyer waits until the first future discount, then the waiting time is geometrically distributed. The expected discounted utility from waiting is
\[
\begin{aligned}
    u_{\mathrm{wait}}
    &=\sum_{\tau=1}^{\infty}
    (1-q)^{\tau-1}q\gamma^\tau(v-(p-\epsilon))\\
    &=\frac{q\gamma}{1-\gamma(1-q)}(v-p+\epsilon).
\end{aligned}
\]
Therefore the buyer purchases at the original price if and only if $u_{\mathrm{buy}}\geq u_{\mathrm{wait}}$, or equivalently
\begin{equation}
    v \geq p+\frac{\gamma}{1-\gamma}q\epsilon.
    \label{eq:warmup-original-threshold}
\end{equation}
Forward-looking behavior therefore acts as an upward shift in the effective purchase threshold at the original price. Buyers with values between $p$ and $p+\frac{\gamma}{1-\gamma}q\epsilon$ would purchase under a constant price $p$, but wait for the lower price in the switchback experiment.

\subsection{Bias of the Naive Two-Price Estimator}
\label{sec: example bias two price}

First suppose the seller can distinguish same-day sales from delayed sales. In this stationary continuum model with an infinite horizon, let $N^{(0)}$ and $N^{(1)}$ denote the deterministic same-day sales masses conditional on posting $p$ and $p-\epsilon$, respectively. By the threshold calculation above,
\[
    N^{(0)}=D\left(p+\frac{\gamma}{1-\gamma}q\epsilon\right),
    \qquad
    N^{(1)}=D(p-\epsilon).
\]
Consistent with common practice, the naive same-day estimator of the demand gradient is a standard difference-in-means:
\[
    \hat d_N(p)=\frac{N^{(0)}-N^{(1)}}{\epsilon}.
\]
Note that the numerator is control minus treatment because treatment lowers
the price by $\epsilon$: dividing the treatment-minus-control
sales difference by the price change $-\epsilon$ gives
$(N^{(0)}-N^{(1)})/\epsilon$.
Here $\hat d_N(p)$ is a deterministic population contrast. In the general model, the price-normalized sample averages in Equation~\eqref{eqn:same-day-agg} and the estimator in Equation~\eqref{eqn:naive-estimator-2} are random because they depend on the realized price assignments and experiment horizon. Their limiting expectations reduce to the quantities used here in the stationary, homogeneous setting. As $\epsilon\to0$, a first-order Taylor expansion gives

\begin{equation}
    \hat d_N(p)
    \to
    \left(1+\frac{\gamma}{1-\gamma}q\right)D'(p).
    \label{eq:warmup-same-day-bias}
\end{equation}

The naive estimator is asymptotically unbiased when the buyers are not patient or not strategic, i.e., when $\gamma = 0$. However, when buyers are patient and strategic, the two-price estimator is multiplicatively biased, unless the seller knows $\gamma$ and can rescale the estimator. The bias results from the underestimation of original-price demand. %
The analysis when the seller observes only total sales follows the same idea, which we defer to \Cref{sec: bias two price}.

\subsection{Why Two Prices Cannot Identify the Gradient}

\label{sec: example two price nonidentification}

The same calculation also explains the two-price nonidentification result. In the infinite-horizon two-price experiment, both the observed same-day sales and total sales are determined by only two demand values:
\[
    D\left(p+\frac{\gamma}{1-\gamma}q\epsilon\right)
    \quad\text{and}\quad
    D(p-\epsilon).
\]

The first term corresponds to the upward-shifted demand at the higher price; the second corresponds to the true demand at the lower price.
Without knowledge of the private patience level $\gamma$, from a two-price experiment, the seller cannot separately determine whether the observed gap in sales is driven by a steep demand curve, a high patience level, or a combination of the two.

For example, a seller observing a given difference
\[
    \Delta =
    D(p-\epsilon)
    -
    D\left(p+\frac{\gamma}{1-\gamma}q\epsilon\right)
\]
can rationalize it with a relatively steep demand curve and low patience, or with a flatter demand curve and high patience. Since the true demand gradient $D'(p)$ changes across these explanations, no estimator based only on two-price data can recover the gradient uniformly when $\gamma$ is unknown, as we also prove in our general model below.

\subsection{How a Third Price Removes the First-Order Bias}

\label{sec: example debias}
The preceding discussion suggests why a third price helps. Consider a three-price switchback experiment with prices
\[
    p,\qquad p-\epsilon,\qquad p-2\epsilon,
\]
where each discounted price is posted with probability $q$ and the original price is posted with probability $1-2q$. With an argument similar to the one above, it is not hard to show that in the infinite-horizon approximation, the effective purchase thresholds are
\[
    p+\frac{\gamma}{1-\gamma}(3q)\epsilon,\qquad
    p-\epsilon+\frac{\gamma}{1-\gamma}q\epsilon,\qquad
    p-2\epsilon,
\]
for the original, first-discount, and second-discount prices respectively.

Let
\[
    \Delta_1=N^{(1)}-N^{(2)},
    \qquad
    \Delta_2=N^{(0)}-N^{(1)}
\]
be the same-day demand differences between adjacent price levels. In the absence of forward-looking behavior, both differences would estimate the same local demand gradient after division by $\epsilon$. With forward-looking behavior, a Taylor expansion gives the following expression as $\epsilon\to 0$.
\[
\begin{aligned}
    \frac{\Delta_1}{\epsilon}
    &\to
    \left(1+\frac{\gamma}{1-\gamma}q\right)D'(p),
    \\
    \frac{\Delta_2}{\epsilon}
    &\to
    \left(1+\frac{\gamma}{1-\gamma}2q\right)D'(p).
\end{aligned}
\]
The difference $\Delta_2-\Delta_1$ therefore isolates the first-order forward-looking distortion in $\Delta_1$, since,
\begin{equation}
\begin{aligned}
    \frac{\Delta_1-(\Delta_2-\Delta_1)}{\epsilon}
    &=
    \frac{2\Delta_1-\Delta_2}{\epsilon}
    \to
    D'(p).
\end{aligned}
    \label{eq:warmup-three-price-estimator}
\end{equation}
This simple calculation is the core intuition behind our positive result.
\Cref{fig:pictorial} presents a pictorial representation of our debiasing argument.

\begin{figure}[htpb]
\centering
\begin{tikzpicture}[scale=0.6]

\draw[->] (0,0) -- (7,0) node[right] {Price \(p\)};
\draw[->] (0,0) -- (0,7) node[above] {Demand \(D\)};

\draw[dashed] (1,0) -- (1,5);
\draw[dashed] (3,0) -- (3,4);
\draw[dashed] (5,0) -- (5,3);

\node[below] at (1,0) {\(p_* - 2\epsilon\)};
\node[below] at (3,0) {\(p_* - \epsilon\)};
\node[below] at (5,0) {\(p_*\)};

\draw[thick, red] (0.5,5.25) -- (6,2.5);

\draw[dotted, line width=.4mm, green] (5,0.5) -- (3,3.16);
\draw[dotted, line width=.3mm] (3,3.16) -- (1,5.83);
\draw[dashed, line width=.4mm, dash pattern=on 3pt off 1pt, green] (3,3.16) -- (1,5);

\draw[line width=1mm, blue] (3,3.16) -- (1,4.16);

\filldraw[black] (1,5) circle (4pt);
\filldraw[black] (5,0.5) circle (4pt);
\filldraw[black] (3,3.16) circle (4pt);

\filldraw[black] (3,4) circle (1pt);
\filldraw[black] (5,3) circle (1pt);
\filldraw[black] (1,5.83) circle (1pt);
\filldraw[black] (1,4.16) circle (1pt);

\draw[dashed] (1,5) -- (-2.5,5);
\draw[dashed] (3,3.16) -- (-2.5,3.16);
\draw[dashed] (1,5.83) -- (-0.5,5.83);
\draw[dashed] (5,0.5) -- (-0.5,0.5);
\draw[dashed] (1,4.16) -- (-2.5,4.16);

\draw[<->] (-0.5,0.5) -- (-0.5,3.16) node[midway,left] {\(\Delta_2\)};
\draw[<->] (-0.5,3.16) -- (-0.5,5.0) node[midway,left] {\(\Delta_1\)};
\draw[<->] (-0.5,5.0) -- (-0.5,5.83) node[midway,left] {\(\Delta_2 - \Delta_1\)};
\draw[<->] (-2.5,4.16) -- (-2.5,5) node[midway,left] {\(\Delta_2 - \Delta_1\)};
\draw[<->] (-2.5,3.16) -- (-2.5,4.16) node[midway,left] {\(\epsilon|\hat{d}(\pstar)|\)};

\end{tikzpicture}
\caption{Visual representation of main debiasing argument. For simplicity of the figure, $\Delta_i$ quantities on the figure represent the absolute value of the $\Delta_i$ quantities in the text. The \textcolor{red}{solid line} depicts the true demand curve, which for simplicity in the figure is depicted as linear, and the small dots on the \textcolor{red}{line} are the true static demands for these price points. The \textcolor{blue}{thick solid line} represents the final estimate of the demand gradient that we uncover. The larger dots represent observed demands which are distorted by the forward-looking behavior, albeit at different intensities. In our model, in the limit where $\delta \to 0$ and $\epsilon \to 0$, the bias of $\Delta_1$ (resp., $\Delta_2$) is exactly $\Delta_2 - \Delta_1$ (resp., $2(\Delta_2 - \Delta_1)$).  Thus, the two \textcolor{green}{dotted lines} represent the biased estimates of the demand gradient that each of $\Delta_1$ or $\Delta_2$ would uncover. We see that the \textcolor{blue}{thick solid line} is parallel to the \textcolor{red}{true demand curve}, thereby uncovering the correct gradient. }
\label{fig:pictorial}
\end{figure}

The general analysis in \Cref{sec:debias} shows that the same correction extends beyond this warm-up setting: buyer patience can be heterogeneous, patience can be correlated with values, demand can vary over time, and the seller need not observe which sales are same-day sales.
\section{Model and Preliminaries}
\label{sec:model}

The warm-up section isolates the main mechanism in a stationary, homogeneous setting. We now introduce the general framework used for the formal results. This framework allows the composition of buyers to vary over time, permits heterogeneous patience levels that may be correlated with values, and allows buyers to hold arbitrary beliefs about post-experimental prices.

Informally, we consider a discrete time setting, indexed by time periods $t = 1, 2, \dots$.  A single seller sells a product to multiple unit-demand buyers who arrive over time.  We let $p_t$ denote the price set by the seller at time $t$.  In the subsequent development, we first describe buyer behavior, and define appropriate estimands of interest to the seller.  We then introduce a natural and common experiment design, the switchback design, and formally define estimators based on the resulting data, and in particular asymptotically unbiased estimators.

\paragraph{Buyer types.}
We model the market as a continuum of infinitesimal, unit-demand buyers. Each buyer is indexed by an \emph{identity} $\omega$, and arrives over discrete time steps. The buyer identity can be mapped to their types, including arrival time, value of the product, patience level, and private belief. Let $t(\omega)\in\{1,2,\ldots\}$ denote the arrival time of buyer identity $\omega$.  We normalize the population so that, in every period $t$, the total mass of buyers arriving at time $t$ is $1$. \footnote{Formally, one may represent buyer identities by a measure space $(\Omega,\mu)$, with $\omega\in\Omega$ and the normalization $\mu(\{\omega:t(\omega)=t\})=1$ for all $t$. We assume $t(\cdot)$ is measurable.}
The other buyer types also remain fixed over the buyer's life in the market:
a \emph{value} $v(\omega)\in[0,V]$ for one unit of the product, and a \emph{patience level} $\gamma(\omega)\in[0,1]$, where larger $\gamma$ corresponds to greater patience in the sense we define below. In the main text, the sequence of joint distributions of values and patience among arriving buyers is fixed, but may vary across periods, which allows demand to be non-stationary. Randomness in the experimental prices and horizon is distinct from this variation in buyer types. In \Cref{apdx: stochastic demand generation model}, we extend the model to stochastic demand shocks that make the sequence of buyer distributions itself random, with the corresponding demand-limit assumptions imposed almost surely.

\paragraph{Equilibrium Behavior} Informally, upon arrival, a buyer can only make one of two choices each period: either buy the product, or wait until the next period.  Since we assume unit demand buyers, upon buying the product the buyer immediately leaves the system. We assume that a buyer with value $v$ who arrives at time $t$ and has patience level $\gamma$ obtains net utility $\gamma^{\tau} (v - p)$ from purchasing the product at price $p$ at time $t + \tau$.  Buyers act to maximize their expected utility, with respect to beliefs they hold about future prices given the price history. A buyer may instead never purchase, obtaining utility zero. Prices are nonnegative, including under every buyer's beliefs about post-experiment prices. Buying at a negative net utility is therefore dominated by never buying. We restrict attention to measurable strategies that purchase only when $p_t\leq v(\omega)$.

Formally, let $h_t$ denote the public history observed through period $t$. In the switchback setting, this includes experiment prices and whether the experiment has ended, as specified in \Cref{ass:rational}. Writing $e_\tau$ for the indicator that the experiment has ended by the time the buyer makes the period-$\tau$ decision, we have $h_t=((p_1,e_1),\ldots,(p_t,e_t))$. Let $\mathcal{H}$ denote the set of all such histories.    Buyer  $\omega$ chooses a {\em strategy} $s_\omega : \mathcal{H} \to \{0,1\}$, where $s_\omega(h_t) = 0$ (resp., $s(h_t) = 1$) if the buyer chooses to wait (resp., purchase) after history $h_t$. Actions are defined while the buyer has not yet purchased. In realized utility and sales expressions, set $s_\omega(h_t)=0$ after the first purchase, so that each buyer contributes at most once.  For an infinite realization of prices and public experiment-status indicators, the utility of strategy $s_\omega$ is:
\begin{equation}
\label{eq:buyer_utility}
 \sum_{t' \geq t(\omega)} \gamma^{t' - t(\omega)} (v - p_{t'}) s_\omega(h_{t'}).
\end{equation}

We make the following assumption on buyers' patience levels.
\begin{assumption}[Bounded patience]
\label{ass:bounded_patientce}
There exists a constant $\Gamma < 1$ such that $\gamma(\omega) < \Gamma$ for all buyers $\omega$.
\end{assumption}
This assumption bounds buyers' patience uniformly away from one. Otherwise, a buyer with patience level $\gamma = 1$ would be willing to wait infinitely long to buy at a lower price.
The common bound controls the effect of future purchase
opportunities uniformly across buyers, which is needed for
our uniform asymptotic results.

In order to choose an optimal strategy, the buyer needs to hold a belief over the future stochastic evolution of prices after any given history $h_t$. Crucially, we assume that any individual buyer's behavior does not influence future prices (a reasonable assumption, since buyers are infinitesimal). Formally, for each buyer identity $\omega$, let $\phi(\omega)$ denote the buyer's belief about the joint evolution of prices and the publicly observed experiment-status indicators.\footnote{One can formalize $\phi$ as a measurable map from the identity space $(\Omega,\mu)$ into the space of probability measures over infinite sequences of prices and experiment-status indicators (endowed with the $\sigma$-algebra generated by cylinder sets).} The buyer's belief after observing a history $h_t$ is obtained by conditioning $\phi(\omega)$ on $h_t$. We fix measurable conditional transition kernels on public histories, including versions at histories having zero probability under $\phi(\omega)$. Optimality below refers to continuation from any history at which the buyer has not yet purchased. (In subsequent sections, we impose additional constraints on beliefs; namely, that they must be consistent with the prices being chosen by the seller within an experiment.)

For the fixed sequence of buyer distributions in the main model,  for buyer $\omega$ arriving at $t=t(\omega)$, an optimal strategy $s_{\omega}^*$ solves an optimal stopping problem under the belief $\phi(\omega)$.  Formally, for a given time $t \geq t(\omega)$, and a given history $h_t$, define the following {\em continuation utility} under a strategy $s_\omega$:
{\small\[
W(h_t ; \omega, s_{\omega}) = \mathbb{E}_{\phi(\omega)}\!\left[\sum_{t' > t} \gamma^{t' - t}(v-p_{t'})s_{\omega}(h_{t'})\,\middle|\,h_t\right].
\]}

Here $W$ is evaluated after forcing the buyer to wait in period $t$, with no prior purchase, and then following $s_\omega$ from period $t+1$ until its first purchase.

For each buyer, nonnegative prices and the option never to purchase imply that the optimal value lies in $[0,V]$. Let $J_\omega(h_t)$ denote this value while the buyer has not yet purchased. Taking public history as the state, its Bellman equation is
\[
\begin{aligned}
W^*_\omega(h_t)&=\gamma\mathbb E_{\phi(\omega)}
[J_\omega(h_{t+1})\mid h_t],\\
J_\omega(h_t)&=\max\{v-p_t,W^*_\omega(h_t)\}.
\end{aligned}
\]
The Bellman operator preserves the space of measurable functions taking values in $[0,V]$ and is a contraction in the supremum norm with modulus $\gamma(\omega)<1$. Its fixed point is the optimal value; a measurable maximizing action exists because there are only two actions \citep{bertsekas2012dynamic}. We select the maximizing action that buys when indifferent.

Under this tie convention, a strategy $s_\omega^*$ is optimal after every history at which the buyer has not purchased if and only if, for all such $h_t$ with $t\geq t(\omega)$, the following conditions hold:
\begin{align*}
v(\omega) - p_t \geq W(h_t ; \omega, s_{\omega}^*) &\implies s_{\omega}^*(h_t) = 1,\\
v(\omega) - p_t < W(h_t ; \omega, s_{\omega}^*) &\implies s_{\omega}^*(h_t) = 0.
\end{align*}

In other words, in equilibrium, the buyer waits if and only if the continuation value exceeds the immediate utility of buying (i.e., stopping).  It is also straightforward to establish existence of an optimal strategy using standard techniques; see, e.g., \citet{bertsekas2012dynamic}.  Note, however, that we actually provide a far simpler characterization of optimal buyer behavior in our experimental setting in our subsequent development.

\subsection{Estimands}

\paragraph{Static Demand.}

For each $t = 1, 2, \ldots$, we define $F_t(\cdot)$ as the cumulative distribution function (CDF) of $v(\omega)$, conditional on $t(\omega) = t$; in other words, this is the CDF of the values of customers who arrive at time $t$ (since a unit mass of buyers arrives at each time $t$). For each $t = 1, 2, \ldots$, we define the {\em static demand} $D_t(p)$, as the mass of customers arriving at time $t$, who would purchase the product if the price were fixed for all time at a level $p$. Since only customers with valuation at or above $p$ would purchase in this case, we have $D_t(p) = 1 - F_t(p^-)$.
We make the following assumptions on the sequence of  static demand functions:

\begin{assumption}[Existence and Differentiability of Limit Aggregate Demand]\label{ass:limit-exists}
  The average static demand converges pointwise: for each fixed price $p$, the following limit exists:
\begin{align}
    D(\price) \triangleq~& \lim_{T\to \infty}\frac{1}{T}\sum_{t=1}^T D_t(\price).
\end{align}
The limit $D(p)$  is a twice differentiable function.
\end{assumption}

\paragraph{Demand Gradient.} For conciseness, we refer to the limit function $D(p)$ as the \emph{average static demand}. The main estimand that the seller is interested in identifying is the gradient of the average static demand at some original price $p_*$, i.e.,
\begin{equation}
    d(\price_*) \triangleq D'(\price_*).
\end{equation}

Throughout the paper,
we refer to the average static demand as the demand and the average static demand gradient as the demand gradient. We assume that both the demand gradient and the demand Hessian are bounded. Moreover, for the construction of a uniformly asymptotically unbiased estimator, we assume that the demand function can be uniformly second-order approximated at $\pstar$.

\begin{assumption}[Uniformly Approximable  Demand]\label{ass:bounded-der}
    There exist constants $\epsilon_0, d_0, \rho$ such that $|d(\pstar)|<d_0$, and for any $p$ in an $\epsilon_0$-neighborhood of $\pstar$: $|d'(p)|<\rho$.
\end{assumption}
For our positive result on an asymptotically unbiased estimator, we need to make further regularity assumptions on the conditional demand gradient function, conditional on the patience level.

We define $F_{t|\gamma}(\cdot)$ as the CDF of $v(\omega)$, conditional on $t(\omega) = t$ and $\gamma(\omega) = \gamma$; this is the conditional CDF of the values of customers arriving at time $t$ who have patience level $\gamma$.
Let $D_{t\mid\gamma}(p) = 1 - F_{t\mid \gamma}(p^-)$ denote  the corresponding conditional demand function.

Let $G_t$ denote the distribution of patience among buyers arriving in period $t$. In \Cref{ass:conditional}, we assume that all $G_t$ are absolutely continuous with respect to a common reference probability measure $\nu$ on $[0,\Gamma]$. Write:
\[ \mu_t=dG_t/d\nu, \]
for their Radon--Nikodym densities. Thus, for every measurable $A\subseteq[0,\Gamma]$,
\[
G_t(A)=\int_A\mu_t(\gamma)\,\nu(\mathrm d\gamma).
\]
For a given demand-generating process, $\nu$ is fixed across periods and experimental designs. It may differ across demand-generating processes and need not be known to the seller. We choose measurable versions of the densities and conditional demands; weighted conditional quantities are taken to be zero wherever $\mu_t(\gamma)=0$.
We make the following regularity assumption on the conditional demands conditioning on the buyer having patience $\gamma$:
\begin{assumption}[Regularity Conditions for Conditional Demand]\label{ass:conditional}

There are constants $H,d_0,\rho,\epsilon_0>0$ and a set $E\subseteq[0,\Gamma]$ with $\nu(E)=1$ such that, for every $\gamma\in E$ and every period $t$, $0\leq\mu_t(\gamma)\leq H$. The conditional demand $D_{t\mid\gamma}$ is twice differentiable on $I=(p_*-\epsilon_0,p_*+\epsilon_0)$, and bounded with
\[
|d_{t\mid\gamma}(p_*)|\leq d_0,
\qquad |d'_{t\mid\gamma}(p)|\leq\rho
\quad (p\in I),
\]
where $d_{t\mid\gamma}=D'_{t\mid\gamma}$. For every $\gamma\in E$ and each fixed price $p$, the weighted time average converges:
\begin{align}
D_\gamma(p)\triangleq
\lim_{T\to\infty}\frac1T\sum_{t=1}^T
D_{t\mid\gamma}(p)\mu_t(\gamma).
\end{align}
The limit $D_\gamma$ is twice differentiable, and, for each fixed $p\in I$,
\begin{align}
d_\gamma(p)\triangleq D'_\gamma(p)
=\lim_{T\to\infty}\frac1T\sum_{t=1}^T
d_{t\mid\gamma}(p)\mu_t(\gamma).
\end{align}
These are pointwise convergence assumptions in price.
\end{assumption}

The functions $D_\gamma$ and $d_\gamma$ are weighted demand densities with respect to $\nu$, so that
\[
\begin{aligned}
D(p)&=\int D_\gamma(p)\,\nu(\mathrm d\gamma),\\
d(p)&=\int d_\gamma(p)\,\nu(\mathrm d\gamma)
\quad(p\in I).
\end{aligned}
\]

Our assumptions are technical and allow a broad set of buyer-type distributions. The assumptions allow homogeneous patience and any stationary patience distribution $G$, including discrete, continuous, or mixed distributions. Values and patience may be correlated, and their joint distribution may vary across periods subject to the common bounds and time-average limits. However, the uniform density bound limits how patience can change over time. For example, this assumption excludes sequences in which all buyers share one patience level in each time period, and that level takes infinitely many distinct values across time horizon. To see this, suppose all buyers arriving in period $t$ share the same patience level $\gamma_t$. To satisfy $\mu_t\leq H$, the common reference probability measure $\nu$ must assign mass at least $1/H$ to each distinct patience level in this sequence. Since $\nu$ has total mass one, the sequence cannot contain infinitely many distinct patience levels.

\paragraph{Revenue Gradient} The seller is also interested in infinitesimal adjustments to the price around $p_*$, so as to improve upon the average static price revenue function, defined as:
\begin{align}
    \rev(\pstar) = \pstar D(\pstar)
\end{align}
For this purpose, the seller is interested in estimating the gradient of the revenue function:
\begin{align}
    \drev(p_*) \triangleq \rev'(\price_*)
\end{align}
Note that the revenue gradient can be derived directly from the demand and the demand gradient.
\begin{align}
    \drev(\pstar) = \pstar\cdot d(\pstar) + D(\pstar).
\end{align}
As will become evident in later sections, identifying the demand function $D(p_*)$ is straightforward. Thus, we will mainly focus on estimating the static demand gradient $d(p_*)$.

\subsection{Switchback Experiment Design}
\label{sec:prelim-experiment}

Trivially, if the number of experimental periods is large enough, the seller can identify the static demand and the static revenue at a price point $p_*$, by simply posting the price and measuring the average mass of purchases at that price point. %
However, to identify the demand gradient (equiv., revenue gradient), the seller needs to experiment with prices and set prices around the price point $p_*$. We assume that at $t=1$ the seller implements a switchback experiment described as follows.

\paragraph{Switchback Experiment} A switchback experiment is characterized by a set of $K+1$ prices $S = ( \pz, p^{(1)}, \dots, p^{(K)})$, with $p^{(i)}\in [0, \bar{p}]$ for some constant $\bar{p}$, together with associated probabilities $\mb{q} = (q_0, q_1, \ldots, q_K)$, such that $\sum_i q_i = 1$. As a convention prices are sorted in decreasing value: $\pz> p^{(1)}>\dots> p^{(K)}$.
In each period during the experiment, the seller chooses a price $p_t$ from the set $S$ according to the probabilities $\mb{q}$. Prices are chosen i.i.d.~each time period during the experiment.

After each experimental period, the experiment ends with probability $\delta$, independently across periods and independently of the experimental price draws. We let $T$ denote the (random) final period of the experiment.  Since we focus on estimating the demand gradient in the limit as $T\to \infty$, we mainly consider the limit where the experiment period is long enough, i.e.\ $\delta\to 0$.

We refer to the tuple $(S, \mb{q}, \delta)$ as the {\em switchback design}.

\paragraph{Local Price Experiments}
Since we are interested in identifying the demand gradient at an original price point $p_*$, we focus mainly on {\em local price experiments} where the prices in $S$ are {\em bounded} within an $\bar{\low}$ range of the original price $\pz=p_*$.
We will refer to $\bar{\low}$ as the {\em price perturbation bound}.
\begin{equation}\label{assumption: eps bound}
   \bar{\epsilon} \triangleq \pz - p^{(K)} = p_*-p^{(K)}.
\end{equation}
\noindent We mainly consider the case where the original price is $\pz$, i.e., alternative prices are ``discounts'' relative to the original price $\pz$. This assumption is without loss of our main results,\footnote{Our positive results continue to hold if one is given more freedom in the choice of experimental price levels.} but matches the reality of price experiments, which are typically implemented via random small discounts, primarily driven by the intuition that customers should not be worse off when treated.
Moreover, for all of our results we will only be analyzing two-price or three-price experiments, i.e., $K = 1$ or $2$, which also matches the small number of distinct discount levels that arise in retail practice \citep{jha2024can} (e.g., 10\% or 20\% discount).

\paragraph{Buyer Beliefs} %
Buyer beliefs must be specified to determine the buyer's optimal strategy under the experiment.  We consider a partial rational expectations model for buyer beliefs, as follows.
After a switchback experiment starts,
we assume the buyers form rational expectations during the experiment, based on the prices in $S$ and associated probabilities $\mb{q}$, and the experiment ending probability $\delta$.
\begin{assumption}[Rational Experimental Price Expectations]\label{ass:rational}
    A buyer's beliefs about future prices are consistent with the price distribution during the experimental period, as described by $(S, \mb{q}, \delta)$.  In addition, buyers observe the end of the experiment.
\end{assumption}

The role of this assumption is to provide a standard rational-expectations benchmark from economic theory: buyers correctly anticipate the distribution of experimental prices, while the model remains otherwise flexible about buyer types and post-experiment beliefs. This benchmark is also natural in modern online markets, where consumers often observe price histories or use price trackers such as Keepa, making experimental price patterns easy to infer. We allow buyers' beliefs about prices after the experiment, as implied by $\phi(\omega)$, to be arbitrary.  %

We briefly discuss these modeling choices.   The buyer should be able to observe the end by noticing a change in experimental prices (e.g.\ a change in the original price of the product). Note that we do not assume buyers form rational expectations about post-experimental prices, since the seller may adjust post-experimental prices according to the experimental results in a manner that is hard for buyers to accurately anticipate. In our analysis we will focus on the limit where the experiment horizon is long enough (i.e.\ $\delta\to 0$). In this regime, the buyers' beliefs about post-experimental prices have a vanishing effect on their strategy within the experiment.

\subsection{Estimators}
Henceforth we suppose that \Cref{ass:rational} applies to the beliefs $\phi(\omega)$ of each buyer identity $\omega$, and let $s_\omega^*$ denote a resulting optimal strategy. For each period $t$, let $C_t$ denote the \emph{total} sales mass in period $t$ (including buyers who arrived earlier and waited), and let $N_t$ denote the \emph{same-day} sales mass in period $t$ (buyers who both arrive and purchase in period $t$). Given the public history $h_t$, define the sets of buyer identities
\begin{align*}
\mathcal{C}_t
&= \{ \omega : t(\omega) \leq t \ \text{and}\ s_{\omega}^*(h_t) = 1\},\\
\mathcal{N}_t
&= \{ \omega : t(\omega) = t \ \text{and}\ s_{\omega}^*(h_t) = 1\}.
\end{align*}
Thus, $\mathcal{C}_t$ is the set of identities that buy at time $t$, and $\mathcal{N}_t$ is the subset that arrived at time $t$ and buys in the same period. We write\footnote{Formally, buyer identities live in a measure space $(\Omega,\mu)$. We assume $\mathcal{C}_t,\mathcal{N}_t \subseteq \Omega$ are $\mu$-measurable, and define $\text{mass}(A):=\mu(A)$.}
\[
C_t := \text{mass}(\mathcal{C}_t), \qquad N_t := \text{mass}(\mathcal{N}_t),
\]
so that $C_1,\ldots,C_T$ is the realized trajectory of total sales over the experiment, and $N_1,\ldots,N_T$ is the realized trajectory of same-day sales.

An  {\em estimator} $\hat{d}(\cdot)$ for the demand gradient $d(p_*)$ is a function of the observed price sequence $(p_1,\ldots,p_T)$, the realized sales mass $\mb{C}=(C_1, \ldots, C_T)$, and, when observed, of the same-day sales mass $\mb{N}=(N_1,\ldots, N_T)$.
An estimator can also depend on the parameters of the local experiment, i.e. $S, \mb{q}, \delta$, but we will suppress this dependence. Moreover, where obvious from context, we will also suppress the functional dependence on $\mb{C}$ and $\mb{N}$, and simply refer to $\hat{d}(p_*)$ as ``the estimator'' at original price $p_*$.

Throughout, we focus on estimators that are unbiased in the limit of a vanishingly small price perturbation bound and ending probability of the experiment, defined as follows.

For uniform asymptotic statements, the demand-generating processes range over a class with common constants $\Gamma,d_0,\rho,\epsilon_0$ and, when \Cref{ass:conditional} is imposed, $H$. The experimental probabilities are fixed and strictly positive. Reference measures $\nu$ may vary within this class.

\begin{definition}[Asymptotically Unbiased Estimator]

An estimator $\hat{d}(p_*)$ is asymptotically unbiased if its limiting expected value approaches $d(p_*)$ uniformly over this class of demand-generating processes: for each $\eta>0$, there exists $\epsilon_1>0$ such that for every price perturbation bound $\epsilon<\epsilon_1$ and every process in the class,

\begin{equation}
        \bigg|\lim_{\delta\to 0}\expect{\hat{d}(p_*)} - d(p_*)\bigg| <\eta.
    \end{equation}
\end{definition}
\section{Characterization of Buyer Behavior}\label{sec:buyer}

The illustrative example in Section~\ref{sec:warmup} shows how forward-looking behavior raises buyers' purchase thresholds. We now characterize this threshold for an arbitrary set of experimental prices, allowing patience to vary across buyers. We first establish the exact characterization for an infinite-horizon experiment, then show that a positive experiment-ending probability introduces an $O(\delta)$ correction in the bias, accommodating arbitrary beliefs about post-experiment prices.

In order to define the strategy of a buyer, it is useful to introduce the following quantities related to the switchback design $(S, \mb{q}, \delta)$. We denote with $P$ a random price drawn from the distribution of experimental prices. Then, $\mcQ(p) := \Pr(P\leq p)$ is the probability of seeing a price no greater than $p$, and $\mcE(p)=\expect{p - P\mid P\leq p}$ is the expected discount conditional on receiving a price no greater than $p$.  Note that $\mcQ(p)\mcE(p) = \expect{(p - P)_+}$.

\Cref{lem:bias when delta equal 0} characterizes the buyer's behavior when the experiment runs for an infinite horizon.
In the limit $\delta=0$, essentially, a buyer with value $v$, when faced with a price $p_t$, will be as if their value is slightly lower $v - \bias$, where the \emph{option value of waiting} term $\bias$ is a function of the price level and the buyer's patience level. Now, for all patience and price levels, the behavior will always be so that buyers seem to have an additive shaded value given by the option value of waiting when deciding whether to accept or not the current price. %
An important aspect of this characterization is that the amount of \emph{value-shading} $\bias$ introduced by the buyer's forward-looking behavior does not depend on the buyer's value $v$.

\begin{lemma}[Bias in the Limit]
\label{lem:bias when delta equal 0}
    Suppose \Cref{ass:rational} holds.
Then in the limit when $\delta= 0$, the buyer accepts a price $p_t$ at period $t$ if and only if  $v \geq p_t + \bias(p_t)$, where $\bias(p_t)$ is the option value of waiting:
    \begin{equation*}
        \bias(p_t) = \frac{\patience }{1-\patience} \, \mcQ(p_t)\mcE(p_t).
    \end{equation*}
\end{lemma}

Our analysis of bias focuses on the limit where $\delta\to 0$. As $\delta\to 0$, the effect of post-experimental belief vanishes. \Cref{corollary: bias in the limit} provides an error quantification of the option value of waiting.

\begin{lemma}
\label{corollary: bias in the limit}
Suppose Assumptions~\ref{ass:bounded_patientce} and \ref{ass:rational} hold.
There is a constant $K<\infty$, independent of the buyer, public history, and ending probability, such that at every ongoing experimental history and for every $\delta\in(0,1)$,
\[
\begin{aligned}
v\geq p_t+\bias(p_t)+K\delta&\implies\text{the buyer buys},\\
v<p_t+\bias(p_t)-K\delta&\implies\text{the buyer waits}.
\end{aligned}
\]
One may take $K=\Gamma\max\{V,\bar p\}/(1-\Gamma)^2$.
\end{lemma}

\Cref{lem:bias when delta equal 0} and \Cref{corollary: bias in the limit} provide an intuitive, albeit crucial result in our analysis, characterizing  buyer behavior that introduces the bias in a standard switchback experiment.
\section{The Insufficiency of Two-Price Experiments}
\label{sec: two_price_est}

We now establish the negative results suggested by the warm-up in
\Cref{sec:warmup}. Recall that we consider the standard two-price switchback design
\[
S=\{p^{(0)},p^{(1)}\}=\{\pstar,\pstar-\epsilon\},
\qquad
\mb q=(1-q,q).
\]
This is a frequently used experimental design in practice and arguably the simplest type of switchback price experiment one would consider, where on each day the seller offers a discount $\epsilon$ relative to the original price of the product, with probability $q$.
Section~\ref{sec:warmup} derived the relevant bias and identification
mechanisms under stationary demand, homogeneous patience, and an infinite
experimental horizon. Below, we prove the result  in our
general framework: demand may vary over time, the experiment has a random
horizon, and the seller may observe either total sales alone or both total
and same-day sales.

\subsection{Bias of Naive Two-Price Estimators}
\label{sec: bias two price}

We first state the result for the naive estimator that one would use if forward-looking behavior were ignored. We start with the case where the seller observes only total sales $\{C_t\}_{t\in[T]}$ on each day.  Define the average number of purchases, properly scaled, within each price level as:
\begin{equation}\label{eqn:total-day-agg}
C^{(i)}=~ \frac{1}{q_i\cdot T}\sum_{t=1}^{T} C_t\cdot \ind{\price_t=p^{(i)}}, \quad
i\in~ \{0, 1\},
\end{equation}
and consider the following canonical estimator $\edd(\pstar)$ of the demand gradient $d(\pstar)$:
\begin{equation}\label{eqn:naive-estimator}
    \edd(\pstar) = \frac{C^{(0)}- C^{(1)}}{\low}.
\end{equation}
As the experiment ending probability $\delta\to 0$, \Cref{thm: multiplicative bias} characterizes the bias of the naive estimator in two-price experiments.

\begin{theorem}[Bias of Two-Price Experiments]
\label{thm: multiplicative bias}
 Suppose that all buyers have the same patience level $\patience$ and Assumptions~\ref{ass:bounded_patientce},\ref{ass:limit-exists},~\ref{ass:bounded-der},~\ref{ass:rational} are satisfied.
 The estimator of the demand gradient $\edd(\pstar)$ defined in Equation~\eqref{eqn:naive-estimator} is multiplicatively biased, with
    \begin{equation*}
    \lim_{\low\to 0}\lim_{\delta\to 0}\expect{\edd(\pstar)}=\left(\frac{1}{q}+\frac{\patience}{1-\patience}\right)\cdot d(\pstar).
\end{equation*}
\end{theorem}

Hence, the naive estimator multiplicatively overestimates the magnitude of the demand gradient due to the forward-looking behavior of customers. The overestimation bias can be decomposed into two sources:
\begin{itemize}
    \item {\bf Bias 1} Underestimation of $D(\pstar)$: when the price is $\pstar$, some buyers with $v\geq \pstar$ wait for a future discount, resulting in an underestimation of $D(\pstar)$.
    \item {\bf Bias 2} Overestimation of $D(\pstar-\low)$: when the price is $\pstar-\low$, the mass of purchases includes buyers who arrived in previous periods, resulting in an overestimation of $D(\pstar-\low)$.
\end{itemize}

The second source of bias can be removed if the seller has the ability to monitor buyers' arrival times. This is  the same case explained by \Cref{sec: example bias two price}, which we state fully generally here.  In that case, the seller can obtain an unbiased estimator of $D(\pstar-\low)$ by removing the mass of buyers who wait and buy and only considering the same-day sales mass $N_t$, instead of $C_t$. More specifically, suppose that the seller observes the mass of purchases $\{\countsameday_t\}_{t\in[T]}$ from same-day buyers. The purchases are aggregated into same-day purchase masses within periods that had the same price level:
\begin{align}\label{eqn:same-day-agg}
\countsameday[i]=\frac{1}{T\lowprob_i}\sum_{t=1}^T \countsameday_t \cdot \ind{p_t=\price[i]}.
\end{align}
We can then consider the corresponding naive estimator $\edd(\pstar)$ of the demand gradient $d(\pstar)$, where we ignore buyers who arrived earlier and waited for the discount:
\begin{align}\label{eqn:naive-estimator-2}
    \edd(\pstar) &= \frac{N^{(0)}- N^{(1)}}{\low}.
\end{align}
In this case, the expectation of the preceding estimator can be shown to be $\left(1+\frac{\patience}{1-\patience}q\right)\cdot d(\pstar)$.

\begin{theorem}[Bias of Two-Price Experiments with Arrival Monitoring]
\label{thm: multiplicative bias-2}
 Suppose that all buyers have the same patience level $\patience$, Assumptions~\ref{ass:bounded_patientce},~\ref{ass:limit-exists},~\ref{ass:bounded-der},~\ref{ass:rational} are satisfied and that the seller observes the same-day sales $\{N_t\}_{t=1}^T$.
 The estimator of the demand gradient $\edd(\pstar)$ in Equation~\eqref{eqn:naive-estimator-2} is multiplicatively biased, with
    \begin{equation*}
    \lim_{\low\to 0}\lim_{\delta\to 0}\expect{\edd(\pstar)}=\left(1 + \frac{\patience}{1-\patience} q\right)\cdot d(\pstar).
\end{equation*}
\end{theorem}
\noindent The proof follows an identical argument as in the proof of \Cref{thm: multiplicative bias} and is omitted. Even though this estimator is still multiplicatively biased, the factor in the bias is a smaller constant, especially for small values of $q$.

\subsection{Nonidentification with Two Prices}

The preceding results might suggest that the seller could remove the bias by rescaling the estimator. This is possible only under restrictive assumptions, such as when all buyers have the same patience level and this patience level is known to the seller. When patience is unknown, the seller cannot fully eliminate the bias with a two-price switchback experiment.

\begin{theorem}[Two-Price Nonidentification]\label{thm:impossible}
    Let $p_*\in(0,V)$. Suppose buyer beliefs satisfy Assumption~\ref{ass:rational}, and the demand-generating process satisfies Assumptions~\ref{ass:bounded_patientce}, \ref{ass:bounded-der}, and~\ref{ass:conditional}. No estimator is asymptotically unbiased. This holds whether the seller observes only total sales $C_t$ in each period or also observes same-day sales $N_t$.
\end{theorem}

The preceding result holds even if the buyers all have the same patience level and the value distribution is the same in all periods. The key issue is that with only two prices, the seller cannot disambiguate between the role of impatience and the role of demand elasticity in driving sales.

Recall that any estimator necessarily depends only on the following two summary statistics of the data-generating process, as explained in \Cref{sec: example two price nonidentification}:
\[
    D\left(\pstar+\frac{\gamma}{1-\gamma}q\epsilon\right)
    \quad\text{and}\quad
    D(\pstar-\epsilon).
\]
We construct two distinct demand-generating processes with different buyer patience levels, such that the demand gradients $d_1(\pstar),d_2(\pstar)$ are different but the observed demand differences $\Delta_1,\Delta_2$ are the same. Intuitively, it is not possible to separately identify demand gradients and patience levels: with a two-price switchback, the seller cannot distinguish buyers who could afford to buy at the original price but choose to wait for a discount from buyers who could not afford the original price and choose to wait.
\section{Debiasing Switchback Experiments with Three Prices}

\label{sec:debias}

In this section, we show that the seller can debias the estimation of the demand gradient by conducting a switchback experiment with three prices. We consider in particular an experiment with equal discounts, i.e., $S=\{p^{(0)}, p^{(1)}, p^{(2)}\}=\{p_*, p_*-\epsilon, p_*-2\epsilon\}$ and associated probabilities $\mb{q} = (q_0, q_1, q_2) = (1 - q_1 - q_2, q_1, q_2)$. We assume $q_0,q_1,q_2>0$ and hold these probabilities fixed as $\epsilon$ and $\delta$ vanish.

For any such switchback design, we construct an asymptotically unbiased estimator, even in the setting where the seller only observes the total sales mass every day $C_t$, but cannot track the buyer's arrival time and hence cannot observe which part of these sales were same-day sales $N_t$, or waiting-buyer sales $C_t-N_t$. Moreover, \emph{we allow for the buyer patience levels to be arbitrarily heterogeneous and correlated with their values} (unlike the negative result in \Cref{thm: multiplicative bias}, which assumed a homogeneous patience level).

We accomplish this result in two steps. First, in \Cref{sec:known arrival time} we construct such an estimator for the simpler setting where the seller observes both $N_t$ and $C_t$ as a stepping stone. Then,
in \Cref{sec:unknown arrival time}, we extend the result to the case when the seller has less information and does not monitor the buyer's arrival time. We construct the unbiased estimator by reducing the problem to the former setting in \Cref{sec:known arrival time}.

\subsection{Asymptotically Unbiased Estimator with Known Arrival Times}\label{sec:known arrival time}

Suppose that the seller observes the mass of purchases $\{\countsameday_t\}_{t\in[T]}$ from same-day buyers. The purchases are aggregated into same-day purchase masses within periods that had the same price level $\countsameday[i]$, as defined in Equation~\eqref{eqn:same-day-agg}.
The asymptotically unbiased estimator of the demand gradient is constructed by a linear combination of same-day purchase masses:
    \begin{equation}\label{eq:estimator-demand-elasticity}
    \begin{aligned}
    \eddunbias(\pstar)
    &= \frac{1}{\epsilon}\bigg[
    \countsameday[1]-\countsameday[2]\\
    &\quad-\frac{\lowprob_2}{\lowprob_1}
    \Big(\countsameday[0]-\countsameday[1]
    -(\countsameday[1]-\countsameday[2])\Big)
    \bigg].
    \end{aligned}
\end{equation}

\begin{theorem}\label{thm:unbiased}
Suppose Assumptions~\ref{ass:bounded_patientce},~\ref{ass:limit-exists},~\ref{ass:bounded-der},~\ref{ass:conditional} and \ref{ass:rational} are satisfied.

For any fixed, possibly time-varying sequence of distributions over the buyers' private type $(v, \gamma)$ and post-experiment price beliefs satisfying these assumptions, the estimator in \Cref{eq:estimator-demand-elasticity} is asymptotically unbiased.
\end{theorem}

\Cref{eq:estimator-demand-elasticity} corrects the bias in $\big(\countsameday[1]-\countsameday[2]\big)$ by the ``difference-of-differences'' $\Big(\countsameday[0]-\countsameday[1]
    -(\countsameday[1]-\countsameday[2])\Big)$.
  This ``difference-of-differences'' relies mainly on exploiting different waiting biases at different prices.  %
  With three price levels, there are many ways to estimate the demand gradient via pairwise comparisons of demand across price levels, e.g., across pairs $(p^{(0)}, p^{(1)})$ and   $(p^{(1)}, p^{(2)})$.  Even though all these pairwise estimates are asymptotically biased, they have different limiting biases. After dividing each demand difference by $\epsilon$, the difference between the upper-price and lower-price pairwise estimates has limiting expectation equal to $q_1/q_2$ times the limiting bias of the lower-price pairwise estimate.

  The key to extending this argument beyond the homogeneous-patience example in \Cref{sec: example debias} is that the correction depends only on the known experimental probabilities $q_1$ and $q_2$, but not the unknown patience levels. Although the waiting bias varies with buyer patience, the same correction removes the first-order bias within every patience group. Because the estimator is linear in sales, this cancellation is preserved when sales are aggregated across groups. Consequently, the seller need not observe buyer patience, and the correction remains valid when patience is heterogeneous and correlated with values.

\subsection{Asymptotically Unbiased Estimator with Unknown Arrival Times}
\label{sec:unknown arrival time}

When buyer arrival times are unobserved, the seller can still estimate the demand gradient.
The seller only observes the total sales mass $C_t$ in each period $t$. Using these total sales, the seller can construct asymptotically unbiased estimators $\widehat{\countsameday[i]}$ of the aggregate same-day sales $\countsameday[i]$ for $i=0, 1, 2$:
\begin{equation}\label{eqn:unbiased-same-day}
\begin{aligned}
\widehat{\countsameday[0]}
&\triangleq \frac{1}{T\lowprob_0}\sum_t C_t\ind{p_t=\pz},\\
\widehat{\countsameday[1]}
&\triangleq \frac{1}{T}\sum_t C_t\ind{p_t=\pz}\\
&\quad + \frac{\lowprob_1+\lowprob_2}{T\lowprob_1}
\sum_t C_t\ind{p_t=\price[1]},\\
\widehat{\countsameday[2]}
&\triangleq \frac{1}{T}\sum_t C_t.
\end{aligned}
\end{equation}

\begin{proposition}\label{prop:unbiased-no-arrival-time}
    Under the assumptions of \Cref{thm:unbiased}, for each sufficiently small fixed $\epsilon>0$, the estimators in Equation~\eqref{eqn:unbiased-same-day} are asymptotically unbiased for the expected purchase masses $\expect{\countsameday[i]}$ of same-day buyers, i.e.,
    \[
    \lim_{\delta\to 0}
    \left|
    \expect{\widehat{\countsameday[i]}}
    -\expect{\countsameday[i]}
    \right|=0,\quad i\in \{0, 1, 2\}.
    \]
\end{proposition}

The reconstruction in \Cref{prop:unbiased-no-arrival-time} follows from three observations about purchase behavior in the infinite-horizon limit. The proof in the appendix controls the vanishing effects of a positive ending probability.

\emph{Highest price.} At $\price[0]$, buyers who are willing to buy would already have purchased upon arrival, since every experimental price is at most $\price[0]$. Thus, in the infinite-horizon limit, total sales at this price equal same-day sales. The first estimator in Equation~\eqref{eqn:unbiased-same-day} therefore uses total sales on highest-price days with the usual normalization by $Tq_0$.

\emph{Lowest price.}
Same-day buyers purchase at this price exactly when $v\geq\price[2]$, giving demand $D_t(\price[2])$.  Moreover, every buyer with value above this lowest price eventually purchases as the experiment continues. After averaging over a long experiment, the mass of buyers still waiting at its end is negligible, so total sales per period and same-day demand at the lowest price have the same limit, $D(\price[2])$. This yields the third estimator.

\emph{Middle price.} Total sales at $\price[1]$ include immediate purchases and purchases by buyers who arrived at $\price[0]$ and waited. To see the adjustment, let $n_i$ denote the limiting expected price-normalized same-day sales at price $\price[i]$, and let $c_i$ denote the corresponding limiting expected price-normalized total sales. Buyers who reject $\price[0]$ but accept $\price[1]$ arrive at rate $q_0(n_1-n_0)$; a fraction $q_1/(q_1+q_2)$ of them next encounter $\price[1]$ rather than $\price[2]$. Dividing their middle-price purchases by $q_1$ gives
\[
c_1=n_1+\frac{q_0}{q_1+q_2}(n_1-n_0).
\]
Since $c_0=n_0$ and $q_0+q_1+q_2=1$, rearranging yields $n_1=(q_1+q_2)c_1+q_0c_0$. Replacing $c_0$ and $c_1$ by their observed price-normalized total-sales averages gives the second estimator in Equation~\eqref{eqn:unbiased-same-day}.

Having access to these asymptotically unbiased estimates $\{\hat{\countsameday[i]}\}_{i=0, 1, 2}$, the seller can construct the analogue of the unbiased estimator of $d(\pstar)$ as in \Cref{eq:estimator-demand-elasticity} in a plug-in manner:
\begin{align}
\label{eqn:unbiased-no-arrival}
\hat{d}(\pstar)
&= \frac{1}{\epsilon}\bigg[
\widehat{\countsameday[1]}-\widehat{\countsameday[2]}\\
&\quad-\frac{\lowprob_2}{\lowprob_1}
\Big(
\widehat{\countsameday[0]}-\widehat{\countsameday[1]}
-(\widehat{\countsameday[1]}-\widehat{\countsameday[2]})
\Big)
\bigg].
\end{align}
The following theorem is an immediate corollary of \Cref{thm:unbiased} and \Cref{prop:unbiased-no-arrival-time}, hence we omit its proof.
\begin{theorem}
Suppose Assumptions~\ref{ass:bounded_patientce},~\ref{ass:limit-exists},~\ref{ass:bounded-der},~\ref{ass:conditional} and \ref{ass:rational} are satisfied.

For any fixed, possibly time-varying sequence of distributions over the buyers' private type $(v, \gamma)$ and post-experiment price beliefs satisfying these assumptions, the estimator in \Cref{eqn:unbiased-no-arrival} is asymptotically unbiased.
\end{theorem}

The estimates in Equation~\eqref{eqn:unbiased-same-day} can be expressed in terms of aggregate total sales mass quantities $C^{(i)}$ as defined in Equation~\eqref{eqn:total-day-agg}.
\begin{equation}
\begin{aligned}
\widehat{\countsameday[0]} &= C^{(0)},\\
\widehat{\countsameday[1]}
&\triangleq q_0 C^{(0)} + (q_1+q_2) C^{(1)}\\
&= C^{(1)} - q_0(C^{(1)}-C^{(0)}),\\
\widehat{\countsameday[2]}
&\triangleq q_0 C^{(0)} + q_1 C^{(1)} + q_2 C^{(2)}.
\end{aligned}
\end{equation}
Using this equivalence, in essence, \Cref{prop:unbiased-no-arrival-time} shows that the difference $N^{(1)}-N^{(2)}$ can be replaced by the unbiased estimate $q_2 (C^{(1)} - C^{(2)})$ and the difference $N^{(0)}-N^{(1)}$ can be replaced by the unbiased estimate $(q_1 + q_2) (C^{(0)} - C^{(1)})$. In other words, the aggregate difference in total sales between two price points can be shown to be in expectation a scaled multiple of the aggregate difference of same-day sales of these two price points, where the factor only depends on the discount probabilities. Using the aggregate total sales expressions, the estimate in Equation~\eqref{eqn:unbiased-no-arrival} can be simplified as follows (see Appendix~\ref{app:simple-derivation} for the derivation):
\begin{align}
\label{eqn: unbiased with all data}
\hat{d}(\pstar)
&= q_2 \left(1 + \frac{q_2}{q_1}\right)
\frac{(C^{(1)} - C^{(2)}) - (C^{(0)} - C^{(1)})}{\epsilon}\\
&= q_2 \left(1 + \frac{q_2}{q_1}\right)
\frac{2 C^{(1)} - C^{(0)} - C^{(2)}}{\epsilon}.
\end{align}
Thus, we derived that the unbiased estimator is simply the difference of differences in total sales, appropriately scaled.
When both discounts are given with equal probability $q$, then the scaling factor is $2q$, i.e. the total probability of seeing a discounted price.
\section{Simulation Results}
\label{sec: simulation static parameter}

Our three-price estimator removes strategic bias in the local-price, long-horizon limit but can retain finite-parameter bias and incur higher variance than the naive estimator. In this section, we empirically compare the estimators using a common value of the experimental parameter $\epsilon$. Our comparison reveals that the unbiased estimator outperforms the naive estimator when the buyers are sufficiently patient or demand is sufficiently elastic.

\subsection{Experiment Setup}

In each experiment, we simulate  $500$ periods with $500$ samples of buyer values drawn from a static log-normal distribution at each period and a homogeneous patience level. Since the experiment has a long enough horizon, we assume the buyers behave the same as $\delta = 0$ in the limit in \Cref{lem:bias when delta equal 0} (we note that using this ideal limit behavior can be a poor approximation for large values of $\gamma$ close to $1$).
We generate
 synthetic data for two products, one representing a pack of Diet Cokes and the other a pack of candies, with relatively higher and lower price elasticities, respectively (see \Cref{fig:diet-coke-demand}  in Appendix~\ref{app:simulation} for the demand curves). We pick original prices $\pstar = 8.0$ for both products. We repeat the experiments $1000$ times using the reproducibility package with master seed $20260929$ and report expected bias, standard deviation and RMSE with one Monte Carlo standard error. The standard errors for standard deviation and RMSE use the delta method. Buyers follow the infinite-horizon purchase thresholds throughout the finite observation window; there is no initial backlog or terminal liquidation. Throughout the simulations, $\epsilon$ denotes the adjacent price gap for both two-price and three-price experiments. \Cref{sec:adaptive parameter} compares cases where $\epsilon$ is adaptively selected for each estimator. These adaptive experiments yield the same qualitative findings on the bias-variance tradeoff between the unknown-arrival and naive estimators. %

For Diet Coke with a higher price elasticity,
the parameters of the log-normal distribution are set to $\mu = 2.2$ and $\sigma = 0.4$ to fit a $-1.5$ price elasticity in Brazil \citep{venson2023price} and to match real-world demand curves \citep{lusk2004choice}. The resulting ground-truth demand gradient is $D'(\pstar) = -0.119$.
For a pack of candies, we set the parameters of the log-normal distribution to $\mu = 2.5$ and $\sigma = 1.1$ to fit a $-0.52$ price elasticity \citep{andreyeva2010impact}. The resulting ground-truth demand gradient is $D'(\pstar) = -0.042$. %

\begin{figure}
\vspace{-3mm}
    \centering
    \includegraphics[width=\linewidth]{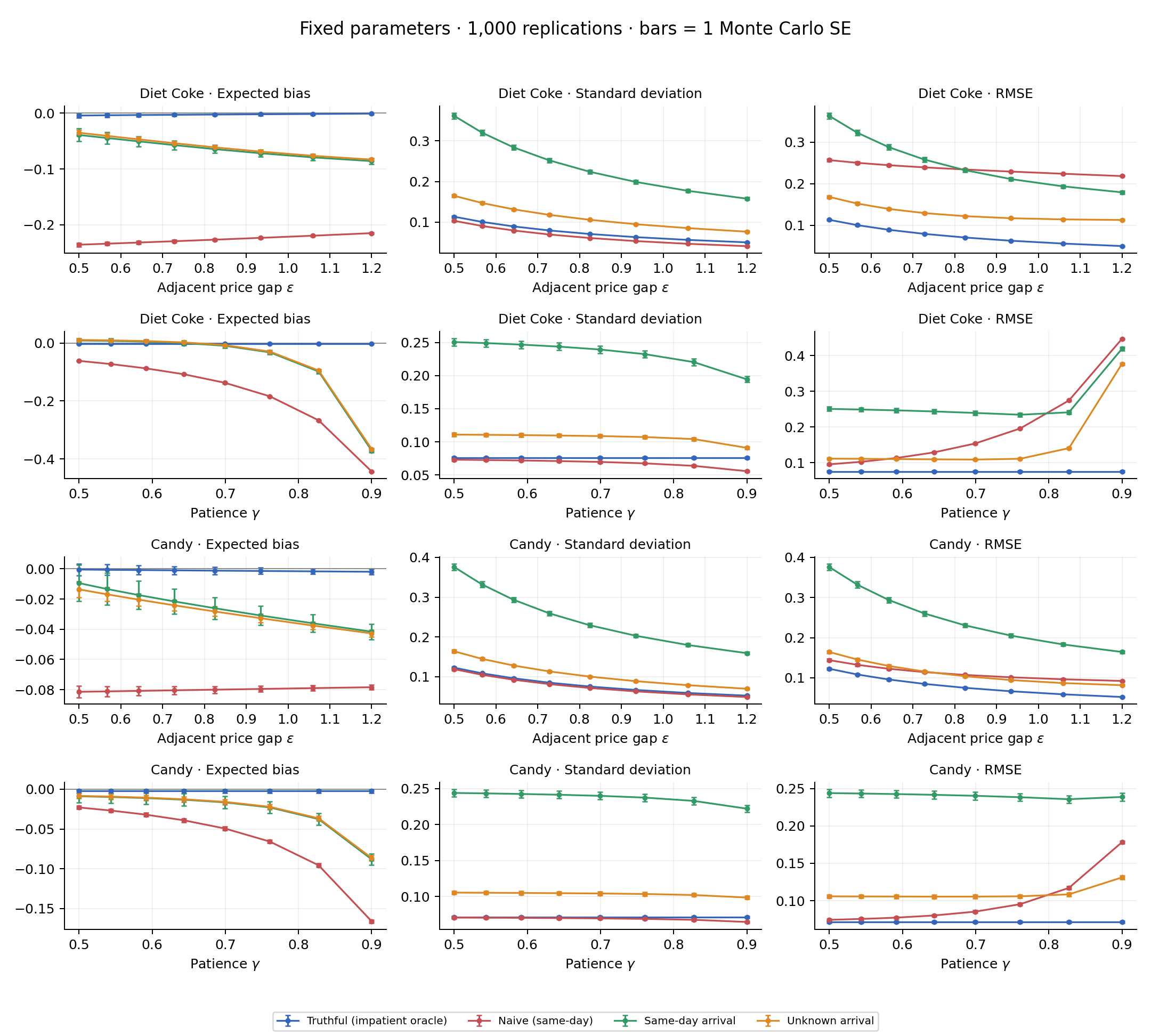}

    \caption{Estimator performance in the seeded reproduction. Rows show Diet Coke with varying $\epsilon$ at $\gamma=0.8$, Diet Coke with varying $\gamma$ at $\epsilon=0.8$, and the corresponding two sweeps for candy. Columns show expected bias, standard deviation, and RMSE. The legend identifies the hypothetical \textcolor{estTruthful}{truthful impatient estimator}, \textcolor{estNaive}{the naive same-day estimator}, \textcolor{estSameDay}{the three-price same-day arrival estimator}, and \textcolor{estUnknown}{the three-price unknown-arrival estimator}. Bars depict one Monte Carlo standard error across $1000$ replications.}
    \label{fig:diet coke-estimate}
    \vspace{-3mm}
\end{figure}

The results from different estimators are visualized in \Cref{fig:diet coke-estimate}.  Numerical description of MSEs is deferred to  \Cref{tab:results-diet-coke} in \Cref{apdx: tab MSE}.
We report the following estimators with one standard error across our $1000$ simulated experiments:  The ground truth demand gradient $D'(\pstar)$, without confidence interval. \textbf{\textcolor{estTruthful}{Truthful Est (impatient):}} we construct a \textbf{hypothetical} naive estimator as in \Cref{eqn:naive-estimator}, where all buyers are impatient ($\patience = 0$) and hence truthfully purchase whenever $\price\leq v$. The hypothetical naive estimator quantifies the unavoidable sampling error in the demand-generating process. \textbf{\textcolor{estSameDay}{Same-day Arrival Est:}} the unbiased estimator in \Cref{eq:estimator-demand-elasticity} with purchase data from same-day buyers. \textbf{\textcolor{estUnknown}{Unknown Arrival Est:}} the unbiased estimator in \Cref{eqn: unbiased with all data} with all purchase data. \textbf{\textcolor{estNaive}{Naive Est:}} the multiplicatively biased naive estimator in \Cref{eqn:naive-estimator-2}.
Our simulation is run for two sets of parameters: 1) fixing $\patience=0.8$ and changing $\epsilon$, and 2) fixing $\epsilon=0.8$ and changing $\patience$. For two-price and three-price experiments, the prices are posted uniformly at random. %

\subsection{Results}
\begin{itemize}
    \item Comparison between unbiased estimators: unknown arrival time estimator outperforms. The estimator with an unknown arrival time consistently exhibits lower variance than the estimator with a known arrival time. By construction, the estimator with an unknown arrival time uses samples from buyers who wait and purchase at a later period, resulting in lower sampling error. However, the favorable performance of our unknown-arrival estimator in the simulations does not establish that arrival information is unhelpful. Whether and how combining arrival information with total sales can improve estimation precision remains an interesting question for future research.
    \item Comparison between unbiased vs.\ biased naive estimator:   unbiased estimator wins when bias dominates or demand is elastic. While the naive estimator suffers a multiplicative bias, our unbiased estimator induces a higher sampling error. Note that the standard deviation of the \textcolor{estNaive}{naive estimator} is comparable to the hypothetical \textcolor{estTruthful}{truthful estimator}, while our \textcolor{estUnknown}{unbiased estimator} has a higher standard deviation. At $\gamma=0.8$, the unknown-arrival estimator has lower RMSE than the naive estimator at every tested $\epsilon\in[0.5,1.2]$ for Diet Coke. Its bias is smaller in magnitude, but is not negligible: it ranges from approximately $-0.035$ to $-0.083$, compared with the true gradient $-0.119$. For candy at $\gamma=0.8$, the unknown-arrival estimator has lower empirical RMSE at the tested gaps $\epsilon\geq0.8246$, while the naive estimator has lower empirical RMSE at the smaller tested gaps. These finite-parameter comparisons do not imply exact unbiasedness.
    \item In settings where the multiplicative bias makes less difference, the naive estimator outperforms our unbiased estimator in RMSE. This occurs in the following scenarios: \textbf{(i)} For candy, the smaller demand-gradient magnitude reduces the cost of multiplicative bias, so the naive estimator can have lower RMSE; this occurs at the smaller tested gaps when $\gamma=0.8$, and at the tested patience levels up to $\gamma\approx0.7609$ when $\epsilon=0.8$.
    \textbf{(ii)} In cases where the buyers are sufficiently impatient, the naive estimator also outperforms the unbiased estimator. Note that the multiplicative factor in bias of a naive estimator is linearly related to $\frac{\patience}{1-\patience}$. For example, at $\epsilon=0.8$ and $\gamma=0.5$ for Diet Coke, the naive estimator has MSE $0.0091$, compared with $0.0124$ for the unknown-arrival estimator. With uniform two-price assignment ($q=1/2$), the local-limit multiplier for the same-day naive estimator is $1+q\gamma/(1-\gamma)=1.5$; the relative bias is therefore $0.5$.
\end{itemize}

Taken together, these findings are useful for practitioners that wish to minimize mean squared error, rather than focus on debiasing as the primary goal.  In particular, the simulations suggest using {\em pilot experiments} to assess the practical cost and benefit of adding a third price. Under comparable market conditions, pilots of two- and three-price designs can estimate each estimator's sampling variance and the difference between their gradient estimates. Under our local-price, long-horizon approximation, the latter provides a quantification of the bias removed by the three-price correction. If this difference is economically meaningful, the pilot's variance estimates can help sellers plan a three-price experiment long enough to estimate the corrected gradient with useful precision. A small but imprecisely estimated difference instead indicates that more information is needed to assess the importance of the correction. Pilots can therefore help plan the experiment's duration and assess the relevance of bias correction.
\section{Conclusion}
\label{sec:conclusion}

In this paper, we study switchback price experiments in the presence of forward-looking buyers. When buyers anticipate future discounts, they may strategically delay purchases, shifting demand away from periods with higher prices and toward periods with lower prices. As a result, the demand observed during a switchback experiment differs from the static demand that would arise if the seller held the price fixed. This creates a form of intertemporal interference that can bias standard estimators of the demand gradient.

We show that the naive estimator based on a two-price switchback experiment is multiplicatively biased, and that this bias can be removed by simply adding one more price into the experiment. Observing buyer arrival times reduces the bias with two prices but does not eliminate the distortion caused by buyers who wait. More fundamentally, when buyer patience is unknown, the static demand gradient is not identifiable from two-price experimental data, even if the seller observes both same-day and delayed sales. We then show that a minimal three-price experiment resolves this problem.  These results provide an actionable approach for practitioners to achieve unbiased gradient estimation from price experiments, an essential step towards using such experiments for revenue optimization.  %

Our numerical results also illustrate the bias-variance tradeoff behind this correction. The three-price estimator removes the bias caused by forward-looking behavior, but can have higher variance than the naive two-price estimator. Our simulation confirms that the unbiased estimator is better when buyers are sufficiently patient or when the magnitude of the demand gradient is sufficiently large, so that strategic waiting creates substantial bias. When buyers are relatively impatient, or demand is locally flat, the lower variance of the naive estimator may lead to lower finite-sample error. If the emphasis is on minimum mean squared error estimation rather than debiasing, then these findings motivate pilot studies to guide experimental design. Finally, while our unknown-arrival estimator performs favorably without tracking arrivals, this does not establish that arrival information is unhelpful. Whether combining arrival information with total sales can improve estimation precision remains an interesting question for future research. %

Our analysis uses rational expectations about the distribution of prices during the experiment. This provides a natural benchmark in markets where buyers can observe price histories or use price-tracking technologies, but it also leaves several directions for future work.
First, if buyers do not recognize the experimental pattern and therefore behave myopically, their behavior is equivalent to the case of zero patience, $\gamma=0$. In this case, buyers do not strategically postpone an affordable purchase. The same-day two-price estimator and our three-price estimators are asymptotically unbiased in the local-price, long-horizon limit. The total-sales two-price estimator still has limiting multiplier $1/q$, because buyers who cannot afford the high price can wait for the low price. At a finite price gap, curvature can introduce approximation bias even when $\gamma=0$.
Other departures from rational expectations may have less immediate implications. If buyers systematically misperceive discount probabilities, learn about the experiment over time, or respond differently to salient price changes, the relative waiting distortions across the three price levels may no longer satisfy the characterization of buyer behavior used by our estimator.
Characterizing the weakest restrictions on buyer beliefs under which identification remains possible is an important direction for future research.

\clearpage
\appendix
\section{Additional Simulation Results}

\subsection{Demand Curves in the Simulation}\label{app:simulation}

\begin{figure}[htpb]
    \centering
    \includegraphics[width=0.45\linewidth]{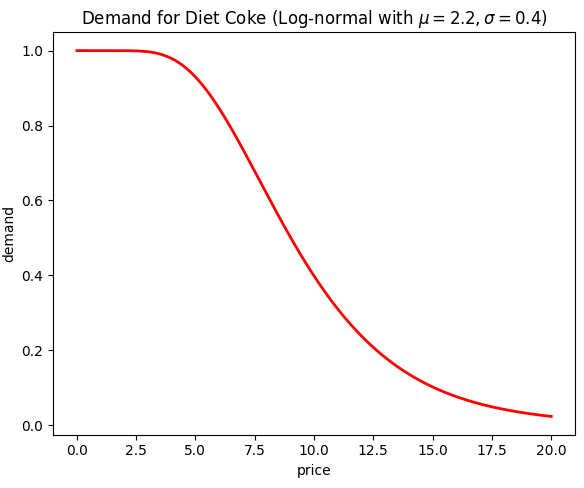}\quad\quad
    \includegraphics[width = 0.45\linewidth]{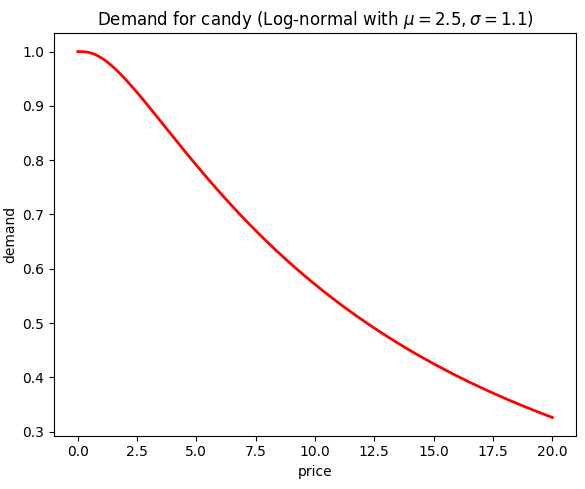}
    \caption{Simulated demand curve for Diet Coke (left) and for a pack of candies (right). }
    \label{fig:diet-coke-demand}
\end{figure}

\subsection{MSE of Different Estimators}
\label{apdx: tab MSE}
\Cref{tab:results-diet-coke} shows MSE of different estimators.
\begin{table}[htpb]
\footnotesize
    \centering
    \begin{tabular}{c|c|c|c|c}
    \hline
    \multicolumn{5}{c}{Diet Coke}\\
    \hline
    Parameters & Truthful & Naive & Same-day & Unknown Arrival\\
    \hline
    $\epsilon=0.5000,\ \gamma=0.8000$ & $0.0128\pm 0.0006$ & $0.0661\pm 0.0016$ & $0.1321\pm 0.0056$ & $\mathbf{0.0284}\pm 0.0012$\\
    $\epsilon=0.8246,\ \gamma=0.8000$ & $0.0050\pm 0.0002$ & $0.0550\pm 0.0009$ & $0.0542\pm 0.0022$ & $\mathbf{0.0149}\pm 0.0006$\\
    $\epsilon=1.2000,\ \gamma=0.8000$ & $0.0025\pm 0.0001$ & $0.0478\pm 0.0006$ & $0.0322\pm 0.0013$ & $\mathbf{0.0128}\pm 0.0005$\\
    \hline
    $\epsilon=0.8000,\ \gamma=0.5000$ & $0.0057\pm 0.0003$ & $\mathbf{0.0091}\pm 0.0004$ & $0.0628\pm 0.0029$ & $0.0124\pm 0.0006$\\
    $\epsilon=0.8000,\ \gamma=0.6996$ & $0.0057\pm 0.0003$ & $0.0237\pm 0.0007$ & $0.0573\pm 0.0026$ & $\mathbf{0.0118}\pm 0.0005$\\
    $\epsilon=0.8000,\ \gamma=0.9000$ & $0.0057\pm 0.0003$ & $0.1997\pm 0.0016$ & $0.1756\pm 0.0047$ & $\mathbf{0.1421}\pm 0.0021$\\
    \hline
    \end{tabular}

\vspace{1mm}
    \begin{tabular}{c|c|c|c|c}
    \hline
    \multicolumn{5}{c}{A pack of candies}\\
    \hline
    Parameters & Truthful & Naive & Same-day & Unknown Arrival\\
    \hline
    $\epsilon=0.5000,\ \gamma=0.8000$ & $0.0150\pm 0.0007$ & $\mathbf{0.0208}\pm 0.0009$ & $0.1412\pm 0.0063$ & $0.0270\pm 0.0012$\\
    $\epsilon=0.8246,\ \gamma=0.8000$ & $0.0057\pm 0.0003$ & $0.0115\pm 0.0004$ & $0.0532\pm 0.0024$ & $\mathbf{0.0108}\pm 0.0005$\\
    $\epsilon=1.2000,\ \gamma=0.8000$ & $0.0027\pm 0.0001$ & $0.0085\pm 0.0003$ & $0.0270\pm 0.0012$ & $\mathbf{0.0067}\pm 0.0003$\\
    \hline
    $\epsilon=0.8000,\ \gamma=0.5000$ & $0.0051\pm 0.0002$ & $\mathbf{0.0055}\pm 0.0002$ & $0.0595\pm 0.0026$ & $0.0112\pm 0.0005$\\
    $\epsilon=0.8000,\ \gamma=0.6996$ & $0.0051\pm 0.0002$ & $\mathbf{0.0073}\pm 0.0003$ & $0.0577\pm 0.0025$ & $0.0111\pm 0.0005$\\
    $\epsilon=0.8000,\ \gamma=0.9000$ & $0.0051\pm 0.0002$ & $0.0317\pm 0.0007$ & $0.0570\pm 0.0024$ & $\mathbf{0.0172}\pm 0.0007$\\
    \hline
    \end{tabular}
    \caption{MSE $\pm$ one Monte Carlo standard error from $1000$ seeded replications. The lowest empirical MSE among the non-oracle estimators is in boldface; boldface does not indicate statistical significance. Parameters are rounded from the geometric grids.}
    \label{tab:results-diet-coke}
    \vspace{-1em}
\end{table}

\subsection{Simulation with Adaptive Parameter Selection}
\label{sec:adaptive parameter}

In this section, we present simulation results with adaptive selection of the $\epsilon$ parameter for each estimator. We use a heuristic selection rule inspired by the Goldenshluger-Lepski method \citep{GL-08, GL-09, GL-11, GL-13}. The simulations evaluate this rule without establishing optimality or confidence coverage.
The algorithm runs with the following steps:
\begin{itemize}
    \item \textit{Initialization}. We select a sequence $\epsilon_0< \dots< \epsilon_{k}$ of $k+1 = 11$ bias parameters for the experiment, with $\epsilon_0 = 0.5$, $\epsilon_k = 2.0$ and each $\epsilon_{i+1}$ multiplicatively larger than $\epsilon_{i}$.
    \item \textit{Estimation}. For each parameter $\epsilon_i$, we simulate $5000$ experiment periods with $200$ samples of buyer values drawn from the same demand distribution as \Cref{sec: simulation static parameter}. We write $\hat{d}_{i}$ as the estimated demand gradient under parameter $\epsilon_i$.
    \item \textit{Adaptive selection}. For each candidate $\epsilon_i$, we estimate the standard deviation $s_i$ of $\hat d_i$ from observable daily estimator contributions. For the truthful and naive estimators, $s_i$ is the sample standard deviation of the daily contributions divided by $\sqrt{T}$. For the unknown-arrival estimator, delayed purchases induce temporal dependence, so we use the sample standard deviation of $100$ nonoverlapping block means, each of length $50$, divided by $\sqrt{100}$. These are plug-in uncertainty estimates. A larger price gap can reduce variance while increasing approximation bias, but neither monotonicity is guaranteed after strategic behavior and adaptive selection. We define a candidate as admissible when its estimate differs from every smaller-gap estimate by at most $1.96s_i$. This is a heuristic comparison threshold, not a confidence interval for the difference between the two estimates.
\begin{align}
    A = \{\low_i: \forall i'<i, |\hat{d}_i- \hat{d}_{i'}|\leq 1.96s_{i}\}.
\end{align}

\item \textit{Output}. The algorithm selects the last $\low_{i^*}$ in the admissible set, i.e.\ $\low_{i^*}=\max A$, with the aim of balancing bias and standard deviation. Every candidate is checked, since the admissible set need not be contiguous. The algorithm outputs the corresponding estimator $\hat{d}_{i^*}$.
\end{itemize}

We simulate $5000$ periods with $200$ samples of buyer values at each period. Each candidate uses this horizon and buyer count. Candidates share price paths within a replication but use independent buyer samples. We use the same adjacent-gap definition of $\epsilon$ for all estimators. The seeded reconstruction corrects an arm-indexing error in the earlier adaptive naive estimator, checks all candidates when taking $\max A$, and uses the uncertainty estimates described above. The earlier implementation used $500$ replications, five candidates in $[0.5,1.5]$, and halved the price gap only for the three-price estimator. The results below therefore replace the historical adaptive figure with a corrected reconstruction using the stated $1000$ replications and $11$ candidates in $[0.5,2.0]$. We report the estimators across $1000$ simulated experiments: \textbf{\textcolor{estTruthful}{Truthful Est (impatient)}},  \textbf{\textcolor{estUnknown}{Unknown Arrival Est}}, and \textbf{\textcolor{estNaive}{Naive Est}}. Because the unknown-arrival estimator has lower empirical MSE than the same-day arrival estimator throughout the fixed-parameter simulations, we omit the latter from this adaptive comparison.

We report the reconstructed results in \Cref{fig:lepski result}. For Diet Coke, the unknown-arrival estimator has lower empirical RMSE than the naive estimator at all eight tested patience levels. For candy, the naive estimator has lower empirical RMSE at the three lowest tested levels ($\gamma=0.5,0.5438,0.5914$), while the unknown-arrival estimator has lower empirical RMSE at the remaining five levels ($\gamma\geq0.6432$). At $\gamma=0.9$, the unknown-arrival and naive RMSEs are respectively $0.2874$ and $0.4610$ for Diet Coke, and $0.1098$ and $0.1578$ for candy. In this reconstruction, the unknown-arrival estimator has smaller absolute bias than the naive estimator at every tested patience level for both products, although substantial finite-parameter bias remains. Its standard deviation is not constant: for Diet Coke it falls from about $0.0407$ at $\gamma=0.5$ to $0.0290$ at $\gamma\approx0.7609$, then rises to $0.0537$ at $\gamma=0.9$; for candy it rises from $0.0274$ to $0.0404$ over the full range. The naive estimator's standard deviation is also not monotone over the entire grid. The high-patience increase in its RMSE is primarily due to increasing absolute bias. These findings depend on the candidate grid and the chosen uncertainty estimates; sensitivity to the block length and the statistical validity of the selection heuristic remain to be assessed.

\begin{figure}
    \centering
    \includegraphics[width=\linewidth]{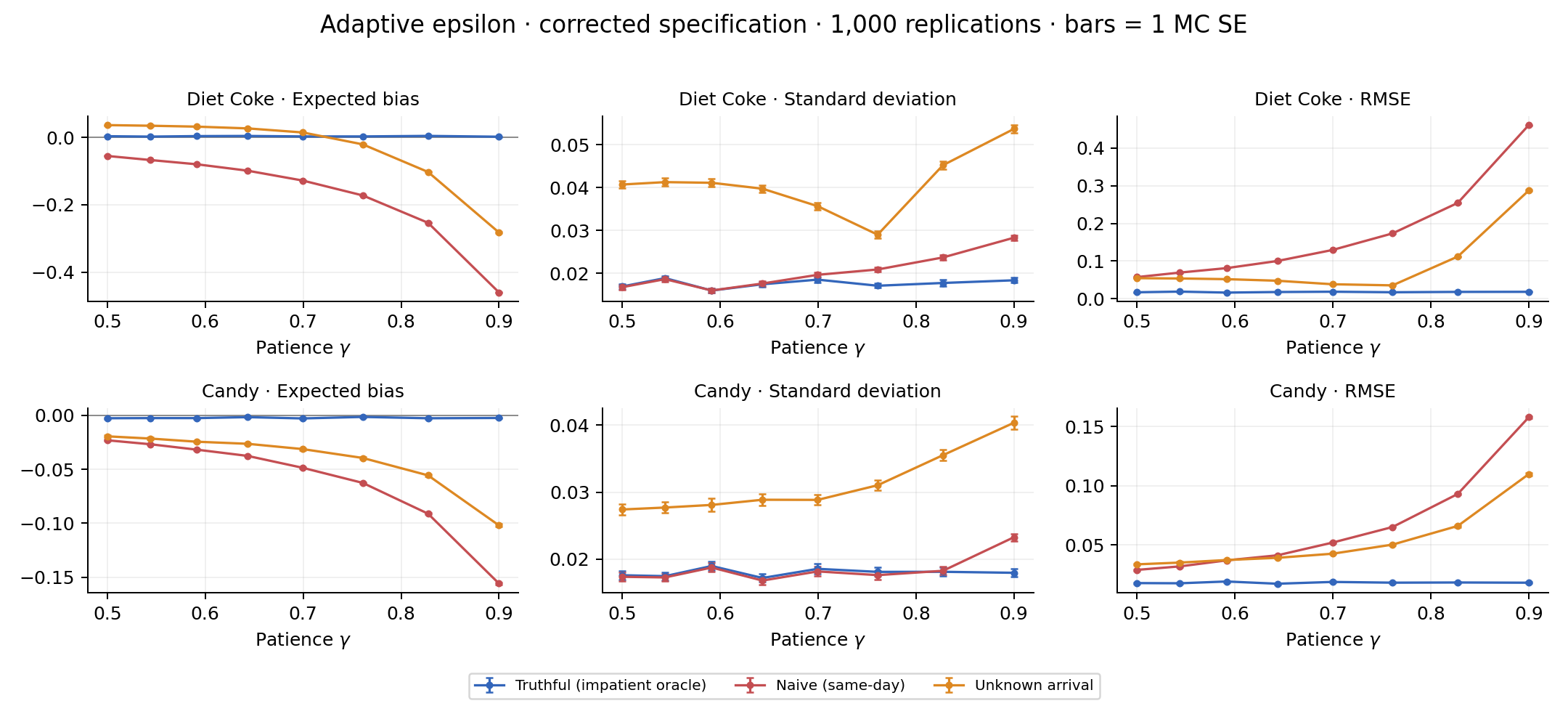}
    \caption{Estimator performance under the corrected adaptive $\epsilon$ selection rule. Each candidate uses $5000$ periods and $200$ arrivals per period, with $11$ geometric adjacent price gaps in $[0.5,2.0]$. Results are based on $1000$ replications; bars depict one Monte Carlo standard error.}
    \label{fig:lepski result}
\end{figure}

\section{Proof of Results in Section~\ref{sec:buyer}}

\subsection{Proof of \Cref{lem:bias when delta equal 0}}

\begin{proof}{Proof of \Cref{lem:bias when delta equal 0}.}
    When $\delta = 0$, the post-experiment belief $\phi(\omega)$ does not matter to a buyer's purchase decisions. Consider a buyer with value $v$ and patience $\patience$. Since the buyer is facing a stationary Markov Decision Process, with a discounted utility, we can write the Bellman equation for the buyer's continuation utility $Z(p_t)$, when observing price $p_t$ at the current period, as:
\begin{equation*}
   Z(p_t)=\max\bigg(v - p_t, \patience\cdot \expect{Z(p_{t+1})}\bigg)
\end{equation*}
By the Contraction Mapping Theorem, the Bellman equation admits a unique solution and there exists an optimal stationary policy that attains the maximum in the equation above \citep{bertsekas2012dynamic}.

From this we observe that during the experimental phase, such optimal policy is given by a threshold price, i.e. the buyer will choose to buy if and only if:
\begin{align*}
    v-p_t \geq \gamma\expect{Z(P)} \Leftrightarrow p_t \leq v - \gamma\expect{Z(P)}.
\end{align*}
Note that the right hand side of the equation above does not depend on $p_t$ since prices are sampled i.i.d. across periods. Hence, the optimal policy is of the form: buyer purchases at time period $t$ if and only if: $p_t\leq \tilde p(v,\gamma)$, for some optimal threshold price $\tilde p(v,\gamma)$, that depends on $v$ and $\gamma$ but not on time period $t$ (for brevity, we will omit these dependencies in the remainder of the proof.).

Now, let $u_{\text{wait}}(\tilde p)$ be the expected continuation utility of the buyer when setting the same threshold price $\tilde p$ at every period of the experimental phase. Hence, the buyer purchases if and only if:
\begin{equation}
\label{eq:waiting_bias_uwait}
    v\geq p_t+u_{\text{wait}}(\tilde p).
\end{equation}

 By the form of the utilities, $u_{\text{wait}}(p)$ takes the form:
\begin{align*}
    u_{\text{wait}}(p)
    =& \expect{\sum_{\tau=1}^{\infty} \patience^\tau (v - p_\tau) \ind{p_\tau \leq  p, \forall \tau'<\tau: p_{\tau'} > p}}
\end{align*}

Note that the
$u_{wait}(p)$ depends only on $\mcE(p) := \expect{p - p_t | p_t \leq p}$ and $\mcQ(p):=\Pr(p_t\leq p)$; both quantities are independent of the time period $t$. In fact, we can write:
    \begin{align}
        u_{wait}(p) =~&\sum_{\tau=1}^{\infty}\patience^\tau  (1-\mcQ(p))^{\tau-1}\mcQ(p)\, (v-p+\mcE(p))\nonumber\\
        =~& \frac{\patience\mcQ(p)}{1-\patience(1-\mcQ(p))}(v-p+\mcE(p)).
    \end{align}

Furthermore, by \eqref{eq:waiting_bias_uwait}, $\tilde p$ must satisfy:
\begin{equation}
    v= \tilde p+u_{\text{wait}}(\tilde p)=\tilde p+\frac{\patience\mcQ(\tilde p)}{1-\patience(1-\mcQ(\tilde p))}(v-\tilde p+\mcE( \tilde p)) .
\end{equation}

Re-arranging and simplifying, the above expression becomes:
 \begin{equation}
     v = \tilde p + \frac{\patience\mcQ(\tilde p)}{1-\patience} \, \mcE(\tilde p)
 \label{eq:waiting bias explicit}
\end{equation}
This equation admits a unique solution because of the monotonicity of the right hand side.

Hence, the buyer purchases if and only if:
 \begin{align*}
      p_t\leq \tilde p = v-\frac{\patience\mcQ(\tilde p)\mcE(\tilde p)}{1-\patience} \ ,
    \end{align*}
    where the last equation follows by \eqref{eq:waiting bias explicit}. Hence, the buyer purchases if and only if:
 \begin{align} \label{eq:buy}
      v\geq p_t+\bias(\tilde p) \ ,
    \end{align}
where $\bias(p)=\frac{\patience\mcQ(p)\mcE(p)}{1-\patience}$. Now, take $p_t\leq\tilde p$, so the buyer purchases. In this case, by \eqref{eq:buy} and because $\bias(p)$ is increasing, we have that if the buyer purchases:
\begin{align*} \label{eq:buy_extended}
      v\geq p_t+\bias(\tilde p)\geq p_t+\bias(p_t)  \ .
    \end{align*}
Take $p_t>\tilde p$, so the buyer does not purchase. In this case, by \eqref{eq:buy} and the monotonicity of $\bias(p)$, we have that if the buyer does not purchase:
\begin{align*}
      v< p_t+\bias(\tilde p)\leq  p_t+\bias(p_t)  \ .
    \end{align*}
 The last two conditions prove the result. \qed

\end{proof}

\subsection{Proof of \Cref{corollary: bias in the limit}}

\begin{proof}{Proof of \Cref{corollary: bias in the limit}.}
Fix a buyer with value $v$ and patience $\gamma$, and an ongoing experimental history $h_t$. Write $W_\delta(v,h_t)$ for optimal continuation utility after waiting in period $t$. For $\beta\in[0,1)$, let $W_\beta^{\infty}(v)$ be the corresponding continuation utility in a perpetual i.i.d.\ experiment with patience $\beta$. The tie convention implies that the buyer buys precisely when
\begin{equation}\label{eq:thre_delta}
v-p_t\geq W_\delta(v,h_t).
\end{equation}
Following a strategy optimal for $W_{\gamma(1-\delta)}^{\infty}(v)$ while the experiment runs, and never buying after it ends, is feasible. Its total expected payoff is exactly $W_{\gamma(1-\delta)}^{\infty}(v)$, since survival for $\tau$ more experimental periods contributes the factor $(1-\delta)^\tau$. Optimal total continuation utility therefore dominates this payoff. Conversely, the experimental component of any feasible strategy's payoff is at most $W_\gamma^{\infty}(v)$: extend the same decisions along ongoing histories in a perpetual experiment, where any additional purchase payoff is nonnegative. If the first post-experiment period is $t+\tau$, the remaining discounted payoff is at most $\gamma^\tau V$. Hence
\begin{equation}\label{eq:experiment ut lowerbound}
\begin{aligned}
W_{\gamma(1-\delta)}^{\infty}(v)
&\leq W_\delta(v,h_t)\\
&\leq W_\gamma^{\infty}(v)
  +\frac{\gamma\delta V}{1-\gamma(1-\delta)}.
\end{aligned}
\end{equation}
The last term uses $\Pr(\text{first post-experiment period}=t+\tau\mid h_t)=\delta(1-\delta)^{\tau-1}$.

To convert these value bounds into decision bounds, put
\[
\begin{aligned}
A(p)&=\expect{(p-P)_+},\\
B_\beta(p)&=\frac{\beta}{1-\beta}A(p),\\
x_\beta(p)&=p+B_\beta(p).
\end{aligned}
\]
Thus $B_\gamma(p)=\bias(p)$. In the perpetual problem, viewed as a function of nonnegative $v$, every stopping rule after waiting has slope in $v$ between zero and $\gamma$. Taking the supremum over stopping rules shows that $W_\gamma^{\infty}$ is nondecreasing and $\gamma$-Lipschitz. Consequently,
\[
g_\gamma(v;p):=v-p-W_\gamma^{\infty}(v)
\]
increases at least at rate $1-\gamma$. By \Cref{lem:bias when delta equal 0}, its unique zero is $x_\gamma(p)$. The upper bound in \eqref{eq:experiment ut lowerbound} therefore implies that the buyer buys whenever
\[
v\geq x_\gamma(p_t)
+\frac{\gamma\delta V}{(1-\gamma)[1-\gamma(1-\delta)]}.
\]
For the converse, buying implies $v-p_t\geq W_{\gamma(1-\delta)}^{\infty}(v)$, so \Cref{lem:bias when delta equal 0} applied with patience $\gamma(1-\delta)$ gives $v\geq x_{\gamma(1-\delta)}(p_t)$. Moreover,
\[
B_\gamma(p)-B_{\gamma(1-\delta)}(p)
=\frac{\gamma\delta A(p)}{(1-\gamma)[1-\gamma(1-\delta)]}
\leq\frac{\Gamma\bar p}{(1-\Gamma)^2}\delta,
\]
because experimental prices lie in $[0,\bar p]$. Both implications in the lemma follow with $K=\Gamma\max\{V,\bar p\}/(1-\Gamma)^2$, uniformly in buyer type and public history. \qed
\end{proof}

\section{Proofs of Results in Section~\ref{sec: two_price_est}}

\subsection{Proof of \Cref{thm:impossible}}
\begin{proof}{Proof.}
Fix $q\in(0,1)$ and $p_*\in(0,V)$. We restrict attention to stationary populations with homogeneous patience, and buyers who believe that the same i.i.d.\ two-price randomization continues after the experiment ends. These beliefs satisfy \Cref{ass:rational}; observing the end does not change future price distributions. Thus \Cref{lem:bias when delta equal 0} applies exactly for every $\delta>0$.

For a population with demand $D$ and patience $\gamma$, define
\[
A_\epsilon=D\left(p_*+\frac{\gamma}{1-\gamma}q\epsilon\right),
\qquad B_\epsilon=D(p_*-\epsilon).
\]
Starting with no buyers before period $1$, at a high price only newly arriving buyers purchase, so $N_t=C_t=A_\epsilon$. At a low price,
\[
N_t=B_\epsilon,\qquad C_t=B_\epsilon+k_t(B_\epsilon-A_\epsilon),
\]
where $k_t$ counts immediately preceding consecutive high-price periods within the experiment. Each such cohort contributes the mass $B_\epsilon-A_\epsilon$ that waits for the first discount. Consequently, conditional on the price path and horizon, all sales observations are determined by $A_\epsilon,B_\epsilon$. Matching these two quantities matches the entire law of $(T,\{p_t,C_t,N_t\}_{t=1}^T)$ for every $\delta>0$.

Choose $0<\gamma_1<\gamma_2<\Gamma$, write $\theta_j=\gamma_j/(1-\gamma_j)$, and choose $a>0$ sufficiently small. Set
\[
b=a\frac{1+q\theta_1}{1+q\theta_2},\qquad 0<b<a.
\]
On a fixed neighborhood $I$ of $p_*$, define
\[
\begin{aligned}
D_1(x)&=\frac12-a(x-p_*),\\
D_{2,\epsilon}(x)&=\frac12+(a-b)\epsilon-b(x-p_*).
\end{aligned}
\]
These local expressions admit smooth extensions with value support in $[0,V]$ and common regularity bounds for all sufficiently small $\epsilon$.\footnote{Choose $r>0$ with $[p_*-2r,p_*+2r]\subset(0,V)$ and $I=(p_*-r,p_*+r)$. Let $\psi$ be nonnegative and smooth, symmetric about $p_*$, supported in $[p_*-2r,p_*+2r]$, and equal to one on $I$; put $J=\int\psi(x)\,\mathrm dx$. Choose smooth probability densities $u_-,u_+$ supported respectively below $p_*-2r$ and above $p_*+2r$, inside $(0,V)$. The density $s\psi+(1-c-sJ/2)u_-+(c-sJ/2)u_+$ has survival function $c-s(x-p_*)$ on $I$. Use $(s,c)=(a,1/2)$ and $(b,1/2+(a-b)\epsilon)$. Small enough $a$ and $\epsilon$ make both tail coefficients nonnegative; the fixed smooth functions and bounded coefficients give common regularity bounds.}
Population $j$ has patience $\gamma_j$, so \Cref{ass:conditional} holds with reference measure $\nu_j$ equal to the point mass at $\gamma_j$ and density $\mu_t=1$. The stationary smooth demands satisfy the remaining regularity and time-average conditions. In particular, $d_1(p_*)=-a$ and $d_{2,\epsilon}(p_*)=-b$.

For sufficiently small $\epsilon$, both prices and both high-price purchase thresholds lie in $I$. Direct calculation gives
\[
\begin{aligned}
A_{1,\epsilon}&=\frac12-aq\theta_1\epsilon,\\
A_{2,\epsilon}&=\frac12+(a-b)\epsilon-bq\theta_2\epsilon
=A_{1,\epsilon},\\
B_{1,\epsilon}&=\frac12+a\epsilon,\\
B_{2,\epsilon}&=\frac12+(a-b)\epsilon+b\epsilon
=B_{1,\epsilon}.
\end{aligned}
\]
The two populations therefore induce identical observation laws at every $\delta>0$. If an estimator were uniformly asymptotically unbiased, its limiting expectations as $\delta\to0$ would exist and coincide; denote their common value by $m_\epsilon$. But
\[
\max\{|m_\epsilon+a|,\ |m_\epsilon+b|\}\geq\frac{a-b}{2}>0.
\]
Both populations belong to the same admissible class with common regularity constants for every sufficiently small $\epsilon$. The second population's dependence on $\epsilon$ is permitted because the definition requires the bias bound to hold uniformly over all processes in the class. The fixed positive lower bound contradicts this definition. Since the construction already has homogeneous patience and stationary values, it also establishes the stated specialization. \qed
\end{proof}

\subsection{Proof of \Cref{thm: multiplicative bias}}

\begin{proof}{Proof.}
Fix $\epsilon>0$ and write
\[
L_\epsilon=p_*-\epsilon,\qquad
x_\epsilon=p_*+\frac{\gamma}{1-\gamma}q\epsilon.
\]
On the same realized experimental price path and finite horizon $T$, let $\bar C_t,\bar N_t$ denote sales when buyers follow the perpetual-experiment policy from \Cref{lem:bias when delta equal 0}. Define $\bar C^{(i)},\bar N^{(i)}$ using the same price normalizations as the actual sales aggregates.

We first justify replacing actual sales by these benchmark sales. For a fixed $h>0$, put
\[
m_t(h)=\sum_{x\in\{L_\epsilon,x_\epsilon\}}
\bigl[D_t(x-h)-D_t(x+h)\bigr].
\]
When $K\delta<h$, \Cref{corollary: bias in the limit} implies that actual and benchmark decisions can differ only for buyers whose values lie within these threshold neighborhoods. Outside them, decisions coincide at every ongoing history, so the first purchase time also coincides. Each affected buyer purchases at most once in either process. Hence, pathwise,
\begin{equation}\label{eq:benchmark-sales-comparison}
\begin{aligned}
|C^{(i)}-\bar C^{(i)}|
&\leq\frac{2}{q_iT}\sum_{t=1}^T m_t(h),\\
|N^{(i)}-\bar N^{(i)}|
&\leq\frac{1}{q_iT}\sum_{t=1}^T m_t(h).
\end{aligned}
\end{equation}
For fixed $h$, the average of $m_t(h)$ converges by \Cref{ass:limit-exists}. Since $T\to\infty$ in probability as $\delta\to0$ and $0\leq m_t(h)\leq2$, bounded convergence gives
\[
\begin{aligned}
\limsup_{\delta\to0}\mathbb E|C^{(i)}-\bar C^{(i)}|
\leq\frac{2}{q_i}\sum_{x\in\{L_\epsilon,x_\epsilon\}}
[D(x-h)-D(x+h)].
\end{aligned}
\]
Letting $h\downarrow0$ and using continuity of $D$ shows that the expected discrepancies vanish; the same conclusion holds for $N^{(i)}$.

In the benchmark, at a high price only new arrivals purchase, and $\bar C_t=D_t(x_\epsilon)$. Independence of prices and $T$ therefore gives
\[
\begin{aligned}
\lim_{\delta\to0}\mathbb E[C^{(0)}]
&=\lim_{\delta\to0}\mathbb E[\bar C^{(0)}]\\
&=\lim_{\delta\to0}\mathbb E\left[\frac1T\sum_{t=1}^T D_t(x_\epsilon)\right]
=D(x_\epsilon).
\end{aligned}
\]
To handle total benchmark sales, let $R_T$ count trailing high-price periods, taking $R_T=T$ if there is no low price. Every buyer with value at least $L_\epsilon$ arriving by the last low-price period purchases by that period. Therefore
\[
0\leq\frac1T\sum_{t=1}^T D_t(L_\epsilon)
-\frac1T\sum_{t=1}^T\bar C_t\leq\frac{R_T}{T}.
\]
Conditional on $T$, independence of prices gives
\[
\begin{aligned}
\mathbb E[R_T\mid T]&=\sum_{r=1}^T(1-q)^r\leq\frac{1-q}{q},\\
\mathbb E[R_T/T]&\leq\frac{1-q}{q}\mathbb E[1/T]\longrightarrow0.
\end{aligned}
\]
It follows from the benchmark comparison and \Cref{ass:limit-exists} that
\[
\lim_{\delta\to0}\mathbb E[(1-q)C^{(0)}+qC^{(1)}]=D(L_\epsilon),
\]
and thus
\[
\lim_{\delta\to0}\mathbb E[C^{(1)}]
=\frac{D(L_\epsilon)}q-\frac{1-q}{q}D(x_\epsilon).
\]

\noindent Thus we conclude that the limit of the estimate of the  demand gradient satisfies,
\[
\begin{aligned}
L
&\triangleq \lim_{\low\to 0}\lim_{\delta\to 0}
\expect{\edd(\pstar)}\\
&=\lim_{\low\to 0}\lim_{\delta\to 0}
\expect{\frac{C^{(0)}-C^{(1)}}{\low}}\\
&=\lim_{\low\to 0} \frac{1}{\low}
\left(D(x_\epsilon)-\frac{D(\pstar-\low)}{\lowprob}
+\frac{1-\lowprob}{\lowprob}D(x_\epsilon)\right)\\
&=\lim_{\low\to 0}\frac{1}{\low q}
\left(D(x_\epsilon)-D(\pstar-\low)\right)\\
&=\lim_{\low\to 0}\frac{1}{\low q}
\left(D'(\pstar)
\left(\frac{\patience}{1-\patience}\lowprob\low+\epsilon\right)
+O(\epsilon^2)\right)\\
&=\left(\frac{1}{\lowprob}
+\frac{\patience}{1-\patience}\right)D'(\pstar).
\end{aligned}
\] \qed
\end{proof}

\section{Proofs of Results in \Cref{sec:debias}}

\subsection{Proof of \Cref{thm:unbiased}}

We prove \Cref{thm:unbiased} by reducing to the case of a homogeneous patience level $\patience$. Subsequently, we will reduce the heterogeneous patience case to the homogeneous patience case, to complete the proof of the theorem.

\begin{lemma}\label{lem:homogeneous-patience-unbias}
    Under the assumptions of \Cref{thm:unbiased}, suppose the buyers have homogeneous patience level $\patience$. The estimator in \Cref{eq:estimator-demand-elasticity} is an asymptotically unbiased estimator of $d(\pstar)$.

    Moreover, there are constants $C,\bar\epsilon_0>0$, common to the admissible demand-generating processes and patience levels, such that for every $\epsilon<\bar\epsilon_0$,
    \begin{equation*}
       \bigg| \lim_{\delta\to 0}\expect{\eddunbias(\pstar)}-d(\pstar)\bigg|\leq C\epsilon.
    \end{equation*}
\end{lemma}

\begin{proof}{Proof.}
We first look at the aggregated purchase masses as $\delta\to 0$.
Let
\[
\begin{aligned}
x_0&=\pstar+\frac{\gamma}{1-\gamma}(\lowprob_1+2\lowprob_2)\low,\\
x_1&=\pstar+\left(-1+\frac{\gamma}{1-\gamma}\lowprob_2\right)\low,\\
x_2&=\pstar-2\epsilon.
\end{aligned}
\]
Also let
\[
\begin{aligned}
b_0&=\frac{\gamma}{1-\gamma}(\lowprob_1+2\lowprob_2),\\
b_1&=-1+\frac{\gamma}{1-\gamma}\lowprob_2,\\
B_\gamma&=\frac{q_2}{q_1}b_0^2
+\left(1+2\frac{q_2}{q_1}\right)b_1^2
+4\left(1+\frac{q_2}{q_1}\right).
\end{aligned}
\]

Shrink the allowable $\epsilon$ uniformly over $\gamma\leq\Gamma$ so that all $x_i$ lie in the common smoothness neighborhood of $p_*$. The comparison argument in \eqref{eq:benchmark-sales-comparison}, now using the three thresholds $x_0,x_1,x_2$, shows that actual and perpetual-policy same-day sales have the same limits as $\delta\to0$. This justifies the benchmark identities below.

\begin{itemize}
	    \item At price $p_t=\pstar$, we have $\mcQ(\pstar)=1$ and $\mcE(\pstar)=(q_1+2q_2)\epsilon$.
	    \begin{align*}
	         \lim_{\delta\to 0}\expect{\countsameday[0]}=~& \lim_{\delta\to 0}\expect{\frac{1}{T\lowprob_0}\sum_t\countsameday_t\ind{p_t=\pstar}}\\
	        =~&  \lim_{\delta\to 0}\expect{\frac{1}{T}\sum_t D_t(x_0)}\\
	        =~& D(x_0)
	    \end{align*}

	    By \Cref{ass:bounded-der}, for some $\rho>0$, there exists $\epsilon_0>0$ such that for each $\epsilon<\epsilon_0$ and every $D$:
\begin{align*}
\bigg|\lim_{\delta\to 0}\expect{\countsameday[0]}
-D(\pstar)
-d(\pstar)b_0\low\bigg|
&\leq \rho(b_0\low)^2.
	    \end{align*}

    \item At price $\pstar-\low$, the buyer believes a lower price $\pstar-2\low$ is posted with probability $\lowprob_2$. Thus following similar reasoning as in the first case, we can argue that:
	    \begin{align*}
\lim_{\delta \to 0} \expect{\countsameday[1]}
&= \lim_{\delta\to 0}
\expect{\frac{1}{T\lowprob_1}
\sum_t\countsameday_t\ind{p_t=\pstar-\epsilon}}\\
&= D(x_1).
	    \end{align*}

	    Again by \Cref{ass:bounded-der}, for  $\rho>0$, there exists the same $\epsilon_0>0$ such that for each $\epsilon<\epsilon_0$ and every $D$:
\begin{align*}
\bigg|\lim_{\delta \to 0} \expect{\countsameday[1]}
-D(\pstar)
-d(\pstar)b_1\low
\bigg|
&\leq \rho b_1^2\epsilon^2.
\end{align*}

	    \item At price $\pstar-2\low$, a buyer with value at least $\pstar-2\low$ will immediately buy. Hence, we can argue that:
	    \begin{align*}
\lim_{\delta\to 0} \expect{\countsameday[2]}
&= \lim_{\delta\to0}\expect{\frac{1}{T\lowprob_2}
\sum_t\countsameday_t\ind{p_t=\pstar-2\epsilon}}\\
&= D(x_2).
	    \end{align*}

    Similarly, for some $\rho>0$, there exists the same $\epsilon_0>0$ such that for each $\epsilon<\epsilon_0$:
    \begin{align*}
        \bigg|\lim_{\delta\to 0} \expect{\countsameday[2]} - D(\pstar) + 2d(\pstar)\epsilon\bigg|\leq4\rho\epsilon^2.
    \end{align*}

\end{itemize}

 By \Cref{ass:bounded-der}, we can write the first-order approximation of $D$:
\[
\begin{aligned}
L_1
&\triangleq
\lim_{\delta\to 0} \expect{(\countsameday[1]-\countsameday[2])}\\
&= D(x_1)-D(x_2)\\
&= d(p_*)\,\left(\epsilon+\frac{\gamma}{1-\gamma}q_2\epsilon\right)
+ O(\epsilon^2),\\
L_2
&\triangleq
\lim_{\delta\to 0} \expect{(\countsameday[0]-\countsameday[1])}\\
&= D(x_0)-D(x_1)\\
&= d(p_*)\,\left(\epsilon
+\frac{\gamma}{1-\gamma}(q_1+q_2)\epsilon\right)
+ O(\epsilon^2)
\end{aligned}
\]

When divided by $\epsilon$, both of these estimands have a first order bias in terms of approximating the derivative $d(p_*)$. However,
notice that the difference-of-differences in demands is a first order proxy for the bias of the first of the two estimands, when appropriately normalized by the known ratio of discount probability $q_2 / q_1$.
\[
\Delta_N=(\countsameday[0]-\countsameday[1])
-(\countsameday[1]-\countsameday[2]).
\]
\begin{align}
\lim_{\delta\to 0} \expect{\Delta_N}
=~& d(\pstar)\cdot\frac{\gamma}{1-\gamma}\lowprob_1\low+O(\low^2),
\end{align}
Let $a_0=-q_2/q_1$, $a_1=1+2q_2/q_1$, and $a_2=-1-q_2/q_1$. Since $\sum_i a_i=0$ and $\sum_i a_i b_i=1$, where $b_2=-2$, the preceding Taylor bounds and the triangle inequality give
\[
\begin{aligned}
\left|\lim_{\delta\to0}\mathbb E[\eddunbias(p_*)]-d(p_*)\right|
&=\left|\frac1\epsilon\sum_{i=0}^2a_iD(x_i)-d(p_*)\right|\\
&\leq\rho\epsilon\sum_{i=0}^2|a_i|b_i^2
=\rho\epsilon B_\gamma.
\end{aligned}
\]
Because the $q_i$ are fixed and positive and $\gamma\leq\Gamma<1$, the constant $z=\sup_{0\leq\gamma\leq\Gamma}B_\gamma$ is finite. If $\bar\epsilon_0>0$ is the common price-gap bound chosen above, then for any tolerance $\eta>0$, taking $\epsilon<\min\{\bar\epsilon_0,\eta/(\rho z)\}$ makes the bias smaller than $\eta$, uniformly over the admissible demand-generating processes and patience levels. \qed
\end{proof}

\begin{proof}{Proof of \Cref{thm:unbiased}.}
All statements below about a patience level hold on the common full-$\nu$-measure set in \Cref{ass:conditional}. Let $\kappa_{s,t}(\gamma;h_t)$ be the fraction of buyers arriving in period $s$ with patience $\gamma$ whose first purchase occurs in period $t$, for the realized public history. Choose measurable versions, setting the fraction to zero when $\mu_s(\gamma)=0$. Define the sales densities with respect to $\nu$ by
\begin{equation}\label{eq:patience-sales-densities}
\begin{aligned}
N_{t,\gamma}&=\mu_t(\gamma)\kappa_{t,t}(\gamma;h_t),\\
C_{t,\gamma}&=\sum_{s=1}^t\mu_s(\gamma)\kappa_{s,t}(\gamma;h_t).
\end{aligned}
\end{equation}
Thus $N_t=\int N_{t,\gamma}\,\nu(\mathrm d\gamma)$ and $C_t=\int C_{t,\gamma}\,\nu(\mathrm d\gamma)$. In particular, delayed purchases retain the density $\mu_s$ of their arrival period.

Let $N_\gamma^{(i)}=(Tq_i)^{-1}\sum_{t=1}^T N_{t,\gamma}\ind{p_t=p^{(i)}}$, and define $C_\gamma^{(i)}$ analogously. Since each buyer purchases at most once, cumulative sales through $T$ are at most cumulative arrivals, which are at most $HT$ in density. Hence, for every realization,
\[
0\leq N_\gamma^{(i)}\leq H/q_i,
\qquad 0\leq C_\gamma^{(i)}\leq H/q_i.
\]
This extends the homogeneous-population argument to cohorts with arrival density bounded by $H$, instead of unit arrival mass.

Write $A_{t,\gamma}(p)=\mu_t(\gamma)D_{t\mid\gamma}(p)$. Then $0\leq A_{t,\gamma}\leq H$, $|A'_{t,\gamma}(p_*)|\leq Hd_0$, and $|A''_{t,\gamma}(p)|\leq H\rho$ on the common neighborhood $I$. Averaging Taylor's formula over $t$ and using the pointwise limits in \Cref{ass:conditional} yields
\[
\left|D_\gamma(p_*+u)-D_\gamma(p_*)-d_\gamma(p_*)u\right|
\leq \tfrac12H\rho u^2
\quad (|u|<\epsilon_0).
\]
The same bounds justify
\[
\begin{aligned}
D(p)&=\int D_\gamma(p)\,\nu(\mathrm d\gamma),\\
d(p)&=\int d_\gamma(p)\,\nu(\mathrm d\gamma)
\quad(p\in I).
\end{aligned}
\]

For fixed $\epsilon>0$ sufficiently small, put $x_i=p_*+b_i\epsilon$, where
\[
b_0=\frac{\gamma}{1-\gamma}(q_1+2q_2),\quad
b_1=-1+\frac{\gamma}{1-\gamma}q_2,\quad b_2=-2.
\]
In the infinite-horizon benchmark of \Cref{lem:bias when delta equal 0}, expected price-normalized same-day sales equal time averages of $A_{t,\gamma}(x_i)$. By \Cref{corollary: bias in the limit}, a positive ending probability changes decisions only for values within $O(\delta)$ of these thresholds, uniformly in public history and buyer patience. The derivative bound on $A_{t,\gamma}$ makes the corresponding discrepancy in expected normalized sales $O(\delta)$, with constants depending on the fixed $q_i$. Since $T\to\infty$ in probability as $\delta\to0$, bounded convergence for these bounded time averages gives
\[
\lim_{\delta\to0}\mathbb E_\delta[N_\gamma^{(i)}]
=D_\gamma(x_i),\qquad i=0,1,2.
\]
This argument requires pointwise time-average limits at the fixed threshold prices, rather than an assumption of uniform convergence in price.

Define $a_0=-q_2/q_1$, $a_1=1+2q_2/q_1$, $a_2=-1-q_2/q_1$, and
\[
\widehat d_\gamma(p_*)=\frac1\epsilon\sum_{i=0}^2a_iN_\gamma^{(i)}.
\]
Direct calculation gives $\sum_i a_i=0$ and $\sum_i a_ib_i=1$ for every $\gamma$. Thus, with $L_\gamma(\epsilon)=\lim_{\delta\to0}\mathbb E_\delta[\widehat d_\gamma(p_*)]$,
\begin{equation}\label{eq:conditional unbiased homogeneous}
|L_\gamma(\epsilon)-d_\gamma(p_*)|
\leq\frac{H\rho\epsilon}{2}\sum_{i=0}^2|a_i|b_i^2
\leq K\epsilon.
\end{equation}
Because $\gamma\leq\Gamma<1$ and the $q_i$ are fixed and positive, one can choose both the allowable $\epsilon$ and $K<\infty$ uniformly over patience levels and demand-generating processes in the specified class.

Linearity gives $\widehat d(p_*)=\int\widehat d_\gamma(p_*)\,\nu(\mathrm d\gamma)$. For each fixed $\epsilon>0$, the pathwise bound
\[
|\widehat d_\gamma(p_*)|
\leq \frac H\epsilon\sum_{i=0}^2\frac{|a_i|}{q_i}
\]
is integrable because $\nu$ is a probability measure. Fubini's theorem and dominated convergence therefore give
\begin{equation}\label{eq:heterogeneous-DCT-1}
\begin{aligned}
\left|\lim_{\delta\to0}\mathbb E_\delta[\widehat d(p_*)]-d(p_*)\right|
&=\left|\int[L_\gamma(\epsilon)-d_\gamma(p_*)]\,\nu(\mathrm d\gamma)\right|\\
&\leq K\epsilon.
\end{aligned}
\end{equation}
This proves uniform asymptotic unbiasedness. For homogeneous patience, $\nu=\delta_{\gamma_0}$ reduces the integral to the single group $\gamma_0$.\qed
\end{proof}

\subsection{Proof of \Cref{prop:unbiased-no-arrival-time}}

\begin{lemma}\label{lem:unbiased-no-arrival-time-homogeneous-patience}
Under the assumptions of \Cref{thm:unbiased}, the estimators in \Cref{eqn:unbiased-same-day} are asymptotically unbiased for expected same-day sales when all buyers have the same patience level $\gamma$. The argument also applies to weighted cohorts whose arrival density is bounded by $H$ and whose weighted demand satisfies the corresponding time-average and derivative conditions.
\end{lemma}

\begin{proof}{Proof of \Cref{lem:unbiased-no-arrival-time-homogeneous-patience}.}
We prove the weighted-cohort version, using $N_{t,\gamma}$, $C_{t,\gamma}$, $N_\gamma^{(i)}$, and $C_\gamma^{(i)}$ from the proof of \Cref{thm:unbiased}. Unit-mass arrivals with homogeneous patience are the special case $\mu_t(\gamma)=1$. Let
\[
\begin{aligned}
y_0&=p_*+\frac{\gamma}{1-\gamma}(q_1+2q_2)\epsilon,\\
y_1&=p_*-\epsilon+\frac{\gamma}{1-\gamma}q_2\epsilon,\\
y_2&=p_*-2\epsilon,
\end{aligned}
\]
and write $n_i=D_\gamma(y_i)$. By the preceding same-day argument, $n_i=\lim_{\delta\to0}\mathbb E_\delta[N_\gamma^{(i)}]$.

First consider the infinite-horizon purchase rule over a deterministic observation horizon $m$. Any buyer willing to purchase accepts the lowest price, which occurs with probability $q_2>0$. The expected density of such buyers still waiting after period $m$ is at most $H\sum_{k\geq0}(1-q_2)^k=H/q_2$. Therefore, counting eventual purchases of cohorts arriving by $m$, instead of purchases completed by $m$, changes any average sales statistic by $O(1/m)$, with constants depending only on $H$ and the fixed probabilities. Averaging the weighted cohort demands then gives the following limiting price-normalized sales identities.

\emph{Highest price.} A buyer who accepts $p^{(0)}$ accepts every experimental price and purchases upon arrival. Thus $c_0=n_0$, where $c_i$ denotes limiting expected price-normalized total sales.

\emph{Lowest price.} Exactly the buyers with value at least $y_2$ eventually purchase under the benchmark rule. Counting them by their arrival cohorts yields
\[
q_0c_0+q_1c_1+q_2c_2=n_2.
\]

\emph{Middle price.} Immediate middle-price purchases contribute $n_1$ after normalization by $q_1$. Additional purchases come from buyers with values in $[y_1,y_0)$ who arrive at the highest price, at limiting density $q_0(n_1-n_0)$. Their first acceptable price is the middle price with probability $q_1/(q_1+q_2)$. Counting eventual purchases and dividing by $q_1$ gives
\[
c_1=n_1+\frac{q_0}{q_1+q_2}(n_1-n_0),
\qquad
q_0c_0+(q_1+q_2)c_1=n_1.
\]
All cohort counts use $\mu_s(\gamma)$ for arrivals in period $s$; no stationarity of these weights is needed.

These identities pass to the actual experiment as $\delta\to0$. Couple its prices with the benchmark price sequence. The threshold perturbation in \Cref{corollary: bias in the limit} can change a buyer's actions only when the value lies within $O(\delta)$ of one of the three thresholds. Uniform local derivative bounds imply that the weighted density of these buyers is $O(\delta)$ per arrival period. Each such buyer contributes at most one purchase under each rule, so expected normalized sales differ by $O(\delta)$. The boundary term vanishes under the geometric horizon since $T\to\infty$ in probability and $\mathbb E_\delta[1/T]\to0$. The bounded weighted time averages also converge in expectation to their pointwise limits.

Consequently, defining
\[
\begin{aligned}
\widehat N_\gamma^{(0)}&=C_\gamma^{(0)},\\
\widehat N_\gamma^{(1)}&=q_0C_\gamma^{(0)}+(q_1+q_2)C_\gamma^{(1)},\\
\widehat N_\gamma^{(2)}&=\sum_{j=0}^2q_jC_\gamma^{(j)},
\end{aligned}
\]
we obtain, for each $i=0,1,2$,
\[
\lim_{\delta\to0}\mathbb E_\delta[\widehat N_\gamma^{(i)}]
=n_i
=\lim_{\delta\to0}\mathbb E_\delta[N_\gamma^{(i)}].
\]
Thus $\lim_{\delta\to0}|\mathbb E_\delta[\widehat N_\gamma^{(i)}]-\mathbb E_\delta[N_\gamma^{(i)}]|=0$.\qed
\end{proof}

\begin{proof}{Proof of \Cref{prop:unbiased-no-arrival-time}.}
For heterogeneous patience, use the sales densities in \Cref{eq:patience-sales-densities}. Linearity and Fubini's theorem give
\[
\begin{aligned}
\mathbb E_\delta[N^{(i)}]
&=\int\mathbb E_\delta[N_\gamma^{(i)}]\,\nu(\mathrm d\gamma),\\
\mathbb E_\delta[\widehat N^{(i)}]
&=\int\mathbb E_\delta[\widehat N_\gamma^{(i)}]\,\nu(\mathrm d\gamma).
\end{aligned}
\]
The preceding lemma makes the difference of these integrands tend to zero for $\nu$-almost every $\gamma$. The pathwise bounds $N_\gamma^{(i)},C_\gamma^{(i)}\leq H/q_i$ imply
\[
\widehat N_\gamma^{(0)}\leq H/q_0,\quad
\widehat N_\gamma^{(1)}\leq H+H(q_1+q_2)/q_1,\quad
\widehat N_\gamma^{(2)}\leq H.
\]
The last inequality follows from the bound $HT$ on cumulative arrivals. These constants dominate the integrands independently of $\delta$ and are integrable under the probability measure $\nu$. Dominated convergence therefore yields
\[
\lim_{\delta\to0}
\left|\mathbb E_\delta[\widehat N^{(i)}]-\mathbb E_\delta[N^{(i)}]\right|=0,
\qquad i=0,1,2.
\]
These are asymptotic identities, with $\epsilon$ fixed while $\delta\to0$.\qed
\end{proof}

\subsection{Derivation of Simplified Version of Estimator with Unknown Arrival}\label{app:simple-derivation}
The simplified formula follows from the following algebraic manipulations:
\[
\Delta_C=(C^{(1)} - C^{(2)}) - (C^{(0)} - C^{(1)}).
\]
Let
\[
R_C=(q_1+q_2)(C^{(0)}-C^{(1)})
-q_2(C^{(1)}-C^{(2)}).
\]
Write $a_q=q_2\left(1+\frac{q_2}{q_1}\right)$. Then
\[
\begin{aligned}
\hat{d}(\pstar)
=~& \frac{1}{\epsilon}\bigg[
q_2(C^{(1)} - C^{(2)})\\
&\qquad
- \frac{\lowprob_2}{\lowprob_1}
R_C\bigg]\\
=~& \frac{1}{\epsilon}\bigg[
q_2\Delta_C\\
&\qquad
- \frac{\lowprob_2^2}{\lowprob_1}
\big(-\Delta_C\big)
\bigg]\\
=~& a_q\frac{\Delta_C}{\epsilon}\\
=~& a_q
\frac{2 C^{(1)}-C^{(0)}-C^{(2)}}{\epsilon}.
\end{aligned}
\]

\section{Applicability to Stochastic Demand-Generating Process}
\label{apdx: stochastic demand generation model}
In this section, we show that our model generalizes to the case where the demand is stochastically generated, under similar assumptions.

The following model allows the sequence of buyer distributions to be random. We state the additional conditions needed to apply the main results conditional on that sequence.

We assume the existence of a (positive-valued) measure space $(\Omega, \mathcal{F}, \mu)$; each $\omega \in \Omega$ represents a buyer type. We make no assumptions on this underlying measure space.

Buyers arrive over time.  In particular, let $t(\omega) \in \{ 1, 2, \ldots\}$ denote the arrival time of a buyer of type $\omega$.  We assume $t$ is measurable w.r.t.~$(\Omega, \mathcal{F}, \mu)$.  Further, we assume that a {\em unit mass} of buyers arrive at each time step, i.e., $\mu(\{ \omega : t(\omega) = t \}) = 1$ for all $t$.
Each arriving buyer has a {\em value} for the product and a {\em patience level} and these quantities remain the same throughout the buyer's life on the market. We allow for aggregate demand uncertainty that affects the value and patience of arriving buyers. We denote with $\{\zeta_{t}\}_{t=1}^{\infty}$ a stochastic process of demand shocks defined on a separate probability space $(\mathcal{Z}, \mathcal{F}', \mathcal{\mu}')$. Let $v(\omega, \zeta_{t(\omega)}) \in [0, V]$ denote the value of a buyer of type $\omega$ (arriving at time $t(\omega)$). Moreover, let $\gamma(\omega, \zeta_{t(\omega)}) \in [0,1]$ denote the {\em patience level} of a buyer of type $\omega$; a larger value represents a more patient buyer, in a sense we make precise below. We assume that $v$ and $\gamma$ are measurable w.r.t.~$(\Omega, \mathcal{F}, \mu)\times (\mathcal{Z}, \mathcal{F}', \mathcal{\mu}')$. We can omit the dependence of $v$ and $\gamma$ on $(\omega, \zeta_{t(\omega)})$.

\paragraph{Static Demand.} {For each $t = 1, 2, \ldots$, we define $F_t[\zeta_t](\cdot)$ as the cumulative distribution function (CDF) of $v(\omega, \zeta_{t(\omega)})$, conditional on $t(\omega) = t$ and demand shock $\zeta_t$; in other words, this is the CDF of the values of customers who arrive at time $t$ (since a unit mass of buyers arrives at each time $t$). Moreover, we use shorthand $F_t(\cdot)$ to denote the random CDF and omit $\zeta_t$. For each $t = 1, 2, \ldots$, we define the {\em static demand} $D_t(p)$, as the mass of customers arriving at time $t$, who would purchase the product if the price were fixed for all time at a level $p$. Since only customers with valuation at or above $p$ would purchase in this case, we have $D_t(p) = 1 - F_t(p^-)$.
Instead of assuming the limit exists, we assume the limit exists almost surely. %
\footnote{As a simple case, if for any $\zeta$, $F_t[\zeta]$ is the same for all $t$ and the demand shocks $\zeta_t$ are i.i.d., then the existence of a limit follows by the strong law of large numbers. However, our assumption is much more permissive and captures many other stochastic demand processes.}
}

{
\begin{assumption}[Existence and Differentiability of Limit Aggregate Demand]\label{ass:limit-exists-stochastic}

On an event of probability one, the average static demand converges for each fixed price $p$:
\begin{align}
    D(\price) \triangleq~& \lim_{T\to \infty}\frac{1}{T}\sum_{t=1}^T D_t(\price) & \text{a.s.}
\end{align}
The limit $D(p)$  is a twice differentiable function.
\end{assumption}
}

Similar as the unconditional daily demand above, we define $F_{t|\gamma}[\zeta_t](\cdot)$ as the CDF of $v(\omega, \zeta_{t(\omega)})$, conditional on $t(\omega) = t$ and $\gamma(\omega) = \gamma$ and demand shock $\zeta_t$; this is the conditional CDF of the values of customers arriving at time $t$ who have patience level $\gamma$.
Write $\zeta=(\zeta_t)_{t\geq1}$ for the full shock realization and $G_{t,\zeta}$ for the patience distribution of period-$t$ arrivals conditional on $\zeta$. Conditional on $\zeta$, experimental prices remain i.i.d. with probabilities $\mathbf q$, and $T$ is an independent geometric horizon with parameter $\delta$. Thus conditioning on demand shocks does not change the experimental randomization law. Buyer beliefs satisfy \Cref{ass:rational} and the main model's nonnegative-price restriction.

On an event of probability one, impose the main regularity assumptions on the realized sequence, with deterministic constants $\Gamma,H,d_0,\rho,\epsilon_0$ common across shock realizations. In particular, let $\nu_\zeta$ be a reference probability measure on $[0,\Gamma]$ such that
\[
G_{t,\zeta}(\mathrm d\gamma)
=\mu_{t,\zeta}(\gamma)\nu_\zeta(\mathrm d\gamma),
\qquad 0\leq\mu_{t,\zeta}(\gamma)\leq H.
\]
For each realization, $\nu_\zeta$ is fixed across periods and experimental designs. Require $\zeta\mapsto\nu_\zeta$ to be a measurable probability kernel, and choose jointly measurable versions of the densities and conditional demands. A deterministic reference measure is a special case; a random but time-constant homogeneous patience level is covered by $\nu_\zeta=\delta_{\gamma_0(\zeta)}$.

For almost every $\zeta$, there is a common full-$\nu_\zeta$-measure set of patience levels on which the local derivative bounds hold for every period and the following limits hold for each fixed price in their stated domains:
\[
\begin{aligned}
D_{\gamma,\zeta}(p)
&=\lim_{m\to\infty}\frac1m\sum_{t=1}^m
\mu_{t,\zeta}(\gamma)D_{t\mid\gamma,\zeta}(p),\\
d_{\gamma,\zeta}(p)
&=D'_{\gamma,\zeta}(p)
=\lim_{m\to\infty}\frac1m\sum_{t=1}^m
\mu_{t,\zeta}(\gamma)d_{t\mid\gamma,\zeta}(p).
\end{aligned}
\]
As in \Cref{ass:conditional}, the first limit is assumed at each fixed price, the second on the common neighborhood of $p_*$, and $D_{\gamma,\zeta}$ is twice differentiable. These are pointwise limits in price. The aggregate demand and gradient, denoted $D_\zeta$ and $d_\zeta$, are obtained by integrating against $\nu_\zeta$.

Conditional on such a shock realization, the arguments for the three-price estimators apply with the same constants. Consequently, for either estimator,
\[
\left|\lim_{\delta\to0}
\mathbb E_{\epsilon,\delta}[\widehat d(p_*)\mid\zeta]
-d_\zeta(p_*)\right|\leq K\epsilon
\quad\text{a.s.}
\]
Here the conditional expectation averages over experimental prices and the horizon. Thus the iterated limit of the conditional expected estimator is the realized gradient $d_\zeta(p_*)$. For fixed $\epsilon$, cumulative purchases are bounded by cumulative arrivals, so the normalized estimators admit deterministic bounds depending only on $\epsilon$ and $\mathbf q$. Dominated convergence over shocks therefore also gives
\[
\lim_{\epsilon\to0}\lim_{\delta\to0}
\mathbb E_{\epsilon,\delta}[\widehat d(p_*)]
=\mathbb E[d_\zeta(p_*)].
\]
The unconditional target is the mean of the random limiting gradient.
The bias formulas
likewise apply conditional on shocks under their respective assumptions.
\end{document}